\documentclass[conference]{IEEEtran}
\IEEEoverridecommandlockouts
\usepackage[left=0.64in, right=0.64in, top=0.73in, bottom=0.5in, includefoot]{geometry}
\usepackage{cite}
\usepackage{array}
\usepackage{algorithm}
\usepackage{algorithmic}

\usepackage{graphicx}
\usepackage{stfloats}
\usepackage{booktabs}
\usepackage{textcomp}

\usepackage{amsthm}

\theoremstyle{plain}
\newtheorem{theorem}{Theorem}
\newtheorem{Definition}{Definition}
\newtheorem{lemma}{Lemma}
\newtheorem{example}{Example}
\newtheorem{remark}{Remark}
\usepackage{amsmath,amssymb,amsfonts}
\usepackage{xcolor}
\def\BibTeX{{\rm B\kern-.05em{\sc i\kern-.025em b}\kern-.08em
    T\kern-.1667em\lower.7ex\hbox{E}\kern-.125emX}}
\usepackage[utf8]{inputenc}   
\usepackage[T1]{fontenc}      
\usepackage[english]{babel} 
\makeatletter
\renewenvironment{proof}[1][\proofname]{%
  \par\pushQED{\qed}\normalfont
  \topsep6\p@\@plus6\p@ \trivlist
  \itemindent\parindent 
  \item[\hskip\labelsep\itshape #1\@addpunct{.}]%
}{%
  \popQED\endtrivlist\@endpefalse
}
\makeatother

\begin{document}

\title{Secure Decentralized Pliable Index Coding for Target Data Size\\
\author{
    \IEEEauthorblockN{Anjali Padmanabhan$^{1}$, Danya Arun Bindhu$^{1}$, Nujoom Sageer Karat$^{1}$, and Shanuja Sasi$^{2}$}
    \IEEEauthorblockA{
        $^{1}$Department of Electronics and Communication Engineering, National Institute of Technology Calicut, India  \\
        $^{2}$Department of Electrical Engineering, Indian Institute of Technology Kanpur, India \\
        E-mail: $\{$anjali\_b220700ec, danya\_b220817ec, nujoom$\}$@nitc.ac.in, shaunjas@iitk.ac.in
    }
}

}

\maketitle

\begin{abstract}
Decentralized Pliable Index Coding (DPIC) problem addresses efficient information exchange in distributed systems where clients communicate among themselves without a central server. An important consideration in DPIC is the heterogeneity of side-information and demand sizes. Although many prior works assume homogeneous settings with identical side-information cardinality and single message demands, these assumptions limit real-world applicability where clients typically possess unequal amounts of prior information. In this paper, we study DPIC problem under heterogeneous side-information cardinalities. We propose a transmission scheme that coordinates client broadcasts to maximize coding efficiency while ensuring that each client achieves a common target level $T$. In addition, we impose a strict security constraint that no client acquires more than the target $T$ number of messages, guaranteeing that each client ends up with exactly $T$ messages. We analyze the communication cost incurred by the proposed scheme under this security constraint.
\end{abstract}

\section{Introduction}
{\it Index coding  (IC)} \cite{bar-yossef2011_indexcoding} problem studies efficient broadcast communication in which a single server transmits coded messages over a noiseless channel to multiple receivers, each having a subset of messages as side-information and demanding an unknown message. By exploiting this side-information, the goal is to minimize the total number of broadcast transmissions\cite{birk1998iscod}. While IC problem traditionally assumes that each receiver demands a predetermined message, many practical systems allow
greater flexibility.  {\it Pliable index coding  (PIC)} \cite{brahma2015_plicable_indexcoding} problem  generalizes this framework by allowing each receiver to decode any message it does not already possess, rather than a predetermined one. Most existing IC and PIC works also assume a centralized transmitter with global knowledge of all messages and side-information; however, such an assumption is often impractical in modern distributed systems \cite{elrouayheb2010cooperative}. This motivates the study of {\it decentralized pliable index coding (DPIC)} \cite{liu2021_decentralized_plicable} problem, where clients themselves generate and broadcast coded transmissions based solely on locally available side-information, typically resulting in an increased number of transmissions due to the lack of global coordination. DPIC problem finds applications in peer-to-peer networks, edge caching, and distributed content sharing platforms, where there is no central transmitter. Another critical aspect is the cardinality of side-information and demand size. Many existing DPIC works assume homogeneous settings where all clients have the same amount of side-information and require exactly one new message. While this simplifies code design, it limits applicability in real-world, heterogeneous scenarios. In \cite{NCCpaper}, DPIC with clients holding heterogeneous side-information cardinalities is studied, and an optimal transmission scheme is proposed that enables all clients to reach a common target knowledge
level.

Beyond efficiency, privacy and security are critical concerns in index coding systems. Prior works have investigated security against external eavesdroppers with or without side-information, introducing notions such as weak, block, and strong security, and characterizing conditions under which linear index codes can prevent leakage of unintended messages\cite{dau2012_security_index_coding,PerfSecureIC,OngSIC}. A complementary line of work studies privacy among receivers, requiring that each user learns nothing beyond its own side-information and intended message, often achieved through shared secret keys between the server and subsets of users\cite{narayanan2018_private_index_coding}. Security considerations have also been explored in the context of PIC problem, where restrictions are imposed to ensure that each client decodes only a single new message, effectively enforcing individual security at the decoding level \cite{sasi2019code}. Security issues become even more pronounced in decentralized settings, where transmissions originate from users themselves rather than a central server. Recent studies on secure DPIC impose individual-security constraints, ensuring that each user decodes only a single new message \cite{liu2020_secure_decentralized_picod}. They demonstrate that decentralization under security constraints can incur a significantly larger transmission overhead compared to centralized settings. Despite these advances, most existing secure DPIC works assume homogeneous side-information sizes and single-message demands. These assumptions simplify the analysis but limit applicability in realistic heterogeneous networks.

\subsection{Our Contributions}
This paper incorporates security considerations into DPIC for scenarios involving clients with non-uniform side-information sizes. We adopt the {\it linearly progressive sets with fixed overlap (LPS-FO)} model from \cite{NCCpaper}, which reflects practical heterogeneity in decentralized networks where clients hold different amounts of side-information, as commonly encountered in intelligent transportation systems. In such systems, vehicles or roadside units exchange traffic-related information and are often pliable in the sense that they are satisfied by receiving any new, useful message. Due to variations in sensing capabilities and route exposure, different vehicles naturally accumulate unequal yet overlapping sets of information, which can be systematically exploited to design efficient decentralized transmission strategies. Under this model, we design a transmission algorithm that schedules client broadcasts in a way that maximizes coding gain per transmission while ensuring all clients reach a common target knowledge level, subject to a  security constraint. The security requirement stipulates that no client should acquire more than a pre-specified target number 
$T$ of messages, that is,  by the end of the transmission phase, each client must possess exactly 
T messages. We also prove that when the number of clients is $3$ or $4$, our proposed scheme is optimal.

\noindent {\it Notations:}
The notation $[a,b]$ represents the set $\{a,a+1,a+2,\ldots, b\}$. The bitwise exclusive OR (XOR) operation is denoted by $\oplus$. For any set $A$, the number of elements in the set (the cardinality of the set) is denoted by $|A|$. For the same set, the $i$-th element is represented by $A[i]$.

\section{System Model}
\label{sec sys model}
We consider a DPIC problem with $M$ messages denoted by $\mathcal{X}=\{x_1,x_2,\ldots,x_M\}$, and $C$ clients denoted by $\{\mathcal{C}_1,\mathcal{C}_2, \dots ,\mathcal{C}_C\}$. Each client $\mathcal{C}_i$ possesses a subset of messages known as side-information set $\mathcal{I}_i \subseteq \mathcal{X}$. 
In this work, we consider  a specific class of heterogeneous side-information structures, which we term as {\it Linearly Progressive Sets with Fixed Overlap (LPS-FO)}.

\begin{Definition}
\label{def LPSFO}
    {\bf Linearly Progressive Sets with Fixed Overlap (LPS-FO)} \cite{NCCpaper}:
The clients' side-information sets are said to be LPS-FO if they satisfy the following two properties:

\begin{itemize}
    \item \textit{Linear Progression:} The side-information sets of the clients are  such that each one has exactly one more message than the previous one and the messages in the side-information sets are consecutive.
    \begin{itemize}
        \item The first client, $\mathcal{C}_1$, starts with $K$ messages: \\$\mathcal{I}_1 = \{x_1, x_2, \dots, x_K\}$.
        \item The size of the side-information sets grows linearly: $|\mathcal{I}_i| = K + (i - 1)$, for any $i \in [1,C]$.
        \item Consequently, the final client, $\mathcal{C}_C$, has the largest set: $|\mathcal{I}_C| = K + (C - 1)$.
    \end{itemize}
    
    \item \textit{Fixed Consecutive Overlap:} Any two consecutive clients share exactly $P$ messages, for any $P \leq K$. Specifically, the last $P$ messages of client $\mathcal{C}_i$ are the first $P$ messages of client $\mathcal{C}_{i+1}$, for $i \in [1,C-1]$. 
\qed
\end{itemize} 

\end{Definition}
Throughout this paper, we assume that the side-information sets are LPS-FO. The clients are indexed sequentially till client $\mathcal{C}_C$  whose side-information set has the largest size. The objective is to raise all clients to a common level of knowledge with security constraints, which is formally defined in Definition \ref{def secure dpic}. 
This target is defined as {one message more than the current holdings of the most knowledgeable client (\(\mathcal{C}_C\))}: $T = |\mathcal{I}_C| + 1 = K + C.$
With security constraint, after the  delivery process, {every client must end up with exactly \( T \) distinct messages}.  
The client with the smallest initial knowledge base, \(\mathcal{C}_1\), begins with only \( K \) messages.  
To reach \( T \), this client must therefore receive {exactly \( C \) new messages}, since
$T - K = (K + C) - K = C .$

\begin{Definition}
\label{def secure dpic}
    {\bf Secure DPIC with LPS-FO}:
    A transmission scheme for DPIC problem with LPS-FO side-information sets is said to be secure if all clients are raised to a common target knowledge level
        $T = |\mathcal{I}_C| + 1 = K + C,$
    and the following condition is satisfied.
    \begin{itemize}
        \item After the delivery process, every client must possess {exactly $T$ distinct messages}. No client acquires more than the target knowledge level 
        $T$. \qed
    \end{itemize} 
\end{Definition}

\begin{theorem}[Lower Bound \cite{NCCpaper}]
\label{thm optimality}
For a DPIC problem with LPS-FO side-information sets, the optimal number of transmissions required to achieve the target knowledge level $T = K + C$ for all clients is  given by $S_{\text{opt}}=C$, if $C \geq 3$,
    and  $S_{\text{opt}}=3$, if $C = 2$.
\end{theorem}
In \cite{NCCpaper}, an optimal transmission scheme was proposed for a case where $K\geq 2P$ which enables each of the clients to reach the target level without security constraint.
 In this paper, we consider secure DPIC with LPS-FO and provide a transmission scheme specifically designed to satisfy security constraints in this model.

\begin{example}
An example for DPIC with LPS-FO side-information structure with an initial side-information set size $K=7$ and a fixed overlap of $P=3$ is shown in Table \ref{tab:sec_example1}. In this example, $M = 75$, the number of clients is $C=9$, and  the set of clients is $\{\mathcal{C}_i: i \in [9]\}$. 
\begin{table}[!t]
\centering
\caption{Clients and side-information sets for example 1 and 2}
\label{tab:sec_example1}
\resizebox{\linewidth}{!}{
\begin{tabular}{|c|c|}
\hline
        Client & Side-information \\
\hline

$\mathcal{C}_1$ & $x_1, x_2, x_3, x_4, x_5, x_6, x_7$ \\

$\mathcal{C}_2$ & $x_5, x_6, x_7, x_8, x_9, x_{10}, x_{11}, x_{12}$ \\

$\mathcal{C}_3$ & $x_{10}, x_{11}, x_{12}, x_{13}, x_{14}, x_{15}, x_{16}, x_{17}, x_{18}$ \\

$\mathcal{C}_4$ & $x_{16}, x_{17}, x_{18}, x_{19}, x_{20}, x_{21}, x_{22}, x_{23}, x_{24}, x_{25}$ \\

$\mathcal{C}_5$ & $x_{23}, x_{24}, x_{25}, x_{26}, x_{27}, x_{28}, x_{29}, x_{30}, x_{31}, x_{32}, x_{33}$ \\

$\mathcal{C}_6$ & $x_{31}, x_{32}, x_{33}, x_{34}, x_{35}, x_{36}, x_{37}, x_{38}, x_{39}, x_{40}, x_{41}, x_{42}$ \\

$\mathcal{C}_7$ & $x_{40}, x_{41}, x_{42}, x_{43}, x_{44}, x_{45}, x_{46}, x_{47}, x_{48}, x_{49}, x_{50}, x_{51}, x_{52}$ \\

$\mathcal{C}_8$ & $x_{50}, x_{51}, x_{52}, x_{53}, x_{54}, x_{55}, x_{56}, x_{57}, x_{58}, x_{59}, x_{60}, x_{61}, x_{62}, x_{63}$ \\

$\mathcal{C}_9$ & $x_{61}, x_{62}, x_{63}, x_{64}, x_{65}, x_{66}, x_{67}, x_{68}, x_{69}, x_{70}, x_{71}, x_{72}, x_{73}, x_{74}, x_{75}$ \\

\hline
\end{tabular}}
\end{table}
This configuration is a valid LPS-FO because it satisfies both defining properties.\\
\textit{Linear Progression:} The size of the side-information sets increases by one with each subsequent client, following $|\mathcal{I}_i| = K + (i-1)$. We have $|\mathcal{I}_1|=7$, $|\mathcal{I}_2|=8$, $|\mathcal{I}_3|=9$, \ldots ,$|\mathcal{I}_9|=15$. The messages in each set are also consecutive. \\
\textit{Fixed Consecutive Overlap:} Any two consecutive clients share exactly $P=3$ messages. Specifically, $\mathcal{I}_1 \cap \mathcal{I}_2 = \{x_5, x_6, x_7\}$, $\mathcal{I}_2 \cap \mathcal{I}_3 = \{x_{10}, x_{11}, x_{12}\}$, \ldots, $\mathcal{I}_8 \cap \mathcal{I}_9 = \{x_{61}, x_{62}, x_{63}\}$.\\
For this example, the target is set to $T = |\mathcal{I}_C| + 1 =  16$. By the end of the delivery phase, each client must possess exactly 16 distinct messages.
\end{example}

 \section{Transmission Scheme for Secure DPIC with LPS-FO Side-information Sets}
 \label{sec trans sch}
 In this section, we present a transmission scheme for a particular class of  secure DPIC with  LPS-FO side-information sets, as specified in Theorem \ref{thm main}. The achievability of this scheme is established via a recursive algorithm that generates the required number of transmissions. This procedure is formally detailed in Algorithm 1.
\begin{theorem}
\label{thm main}
Consider a secure DPIC problem with LPS-FO side-information sets such that $K \geq 2P$, $P \geq r_{\max} - 2$, and $C \geq 3$, where $r_{\max}$ is a positive integer satisfying
\begin{align}
\label{rmax general}
\frac{(r_{\max}-2)(r_{\max}-1)}{2} < C \leq \frac{r_{\max}(r_{\max}-1)}{2}.
\end{align}
For such instances, the recursive transmission scheme described in Algorithm~1 achieves the exact target knowledge level $T = K + C$ for every client. The total number of transmissions required for $C$ clients is given by $N(C)$, where 
\begin{align}
N(C) = C + N(C - r_{\max}),
\end{align}
with  $N(0)=0$, $N(1)=1$, and $N(2)=3$. \qed
\end{theorem}
\begin{table}
\centering
\begin{tabular}{cl}
\toprule
\textbf{Symbol} & \textbf{Description} \\
\midrule
$C$ & Total number of clients in the original problem \\
$C^l$ & Number of clients at recursive level $l$ of \textbf{Algorithm 1} \\
$\mathcal{C}_i^l$ & The $i$-th client at recursive level $l$ \\
$r_{\max}^l$ & Parameter determined by $C^l$ at level $l$ \\
$\mathcal{I}_{F_j}^{l,i}$ & The $j$-th message in the \textit{initial overlapping segment} of \\&side-information set of client $\mathcal{C}_i^l$ \\
$\mathcal{I}_{L_j}^{l,i}$ & The $j$-th message in the \textit{terminal overlapping segment} of \\&side-information set of  client $\mathcal{C}_i^l$ \\
$\mathcal{I}_{U_j}^{l,i}$ & The $j$-th \textit{unique (non-overlapping)} message in the \\&side-informaiton set of  client $\mathcal{C}_i^l$ \\
$\mathcal{I}_{U_L}^{l,i}$ & The last \textit{unique (non-overlapping)} message in the \\&side-informaiton set of  client $\mathcal{C}_i^l$ \\
$W_{j}^{l,i}$ & The $j$-th transmission done by client $\mathcal{C}_i^l$ at recursive level $l$ \\
$\mathcal{W}^l$ & Set of transmissions at level $l$ \\ 
\bottomrule \\
\end{tabular}
\caption{Summary of mathematical notation used in the Algorithms 1-3.}
\label{tab:notations}
\end{table}
\begin{algorithm}
\caption{Main Recursive Algorithm related to Theorem \ref{thm main}}
\label{alg:main-recursive}
\begin{algorithmic}[1]
\REQUIRE Number of clients $C$, client set $\mathcal{C} =\{\mathcal{C}_i, i \in [C]\}$ and their side-information sets with LPS-FO structure.
\ENSURE Complete transmission scheme $\mathcal{W} $

\STATE Initialize: $l = 1$, $C^1 = C, K^{1} = K, \mathcal{W} = \emptyset$
\STATE Initial  set of clients: $\mathcal{C}^1 =\{\mathcal{C}_i^1: \mathcal{C}_i^1=\mathcal{C}_i, i \in [C]\}$

\WHILE{${C}^l > 2$}
    \STATE \textbf{Step 1:} Find the positive integer $r_{\max}^l$ satisfying
    \begin{align}
    \label{rmax}
        \frac{(r_{\max}^l-2)(r_{\max}^l-1)}{2} < C^l \leq \frac{r_{\max}^l(r_{\max}^l-1)}{2}
    \end{align}
    
    \IF{$C^l \neq \frac{(r_{\max}^l-2)(r_{\max}^l-1)}{2} + 1$}
        \STATE $\mathcal{W}^l \gets$ Execute \textbf{Algorithm 2}($C^l, K^l,r_{\max}^l, \mathcal{C}^l$)
    \ELSE
        \STATE $\mathcal{W}^l \gets$ Execute \textbf{Algorithm 3}($C^l, K^l, r_{\max}^l, \mathcal{C}^l$)
    \ENDIF
    
    \STATE Add all transmissions from $\mathcal{W}^l$ to  $\mathcal{W}$
    
    \STATE \textbf{Step 2:} Prepare clients for next recursive level
    \STATE $C^{l+1} = C^l - r_{\max}^l$
    \STATE $K^{l+1} = K^l + r_{\max}^l$
    \IF{$C^{l+1} > 0$}
        \FOR{$i = 1$ \TO $C^{l+1}$}
        \STATE $\quad \mathcal{C}_i^{l+1} \gets \mathcal{C}_{r_{\max}^l + i}^l$ 
        \ENDFOR
    \ENDIF
    
    \STATE $l \gets l + 1$ \COMMENT{Move to next recursive level}
\ENDWHILE

\STATE \textbf{Base Cases:}
\IF{$C^l = 2$}
    \STATE $\mathcal{C}_2^l$ transmits $W_{1}^{l,2} = \mathcal{I}_{F_1}^{l,2} \oplus \mathcal{I}_{U_1}^{l,2}$
    \STATE $\mathcal{C}_2^l$ transmits $W_2^{l,2} = \mathcal{I}_{F_1}^{l,2} \oplus \mathcal{I}_{U_2}^{l,2}$
    \STATE $\mathcal{C}_1^l$ transmits $W_1^{l,1} = \mathcal{I}_{L_1}^{l,1} \oplus \mathcal{I}_{U_1}^{l,1}$
    \STATE $\mathcal{W}^l = \{W_1^{l,2}, W_2^{l,2}, W_1^{l,1}\}$
    \STATE Add $\mathcal{W}^l$ to $\mathcal{W}$
\ELSIF{$C^l = 1$}
    \STATE $\mathcal{C}_{C-1}^1$ transmits $W_1^{l, C-1} = \mathcal{I}_{F_1}^{1,{C-1}} \oplus \mathcal{I}_{U_L}^{1,{C-1}}$
    \STATE $\mathcal{W}^l = \{W_1^{l, C-1}\}$
    \STATE Add $\mathcal{W}^l$ to $\mathcal{W}$
\ENDIF

\RETURN $\mathcal{W}$
\end{algorithmic}
\end{algorithm}
\begin{algorithm}
\caption{General Case Transmission Scheme}
\label{alg:general-case}
\begin{algorithmic}[1]
\REQUIRE $C^l$ clients, $\mathcal{C}^l =\{\mathcal{C}_i^l, i \in [C^l]\}$,  cardinality of side-information set of the first client in  $\mathcal{C}^l$ as $K^l$,  parameter $r_{\max}^l$ such that $C^l \neq \frac{(r_{\max}^l-2)(r_{\max}^l-1)}{2} + 1$
\ENSURE Transmission scheme $\mathcal{W}^l$ for level $l$

\STATE  $\mathcal{W}^l = \emptyset$

\STATE \textbf{Phase 1: Intermediate Clients ($\mathcal{C}_2^l$ to $\mathcal{C}_{r_{\max}^l-2}^l$)}
\FOR{$i = 1$ \TO $r_{\max}^l-3$}
    \STATE $\mathcal{C}_{i+1}^l$ transmits $W_{1}^{l,i+1} = \mathcal{I}_{F_1}^{l,i+1} \oplus \mathcal{I}_{L_1}^{l,i+1}$
    \FOR{$j = 1$ \TO $i-1$}
        \STATE $\mathcal{C}_{i+1}^l$ transmits $W_{j+1}^{l,i+1} = \mathcal{I}_{F_1}^{l,i+1} \oplus \mathcal{I}_{U_j}^{l,i+1}$
    \ENDFOR
\ENDFOR

\STATE Let $T^l = \frac{(r_{\max}^l-2)(r_{\max}^l-3)}{2}$

\STATE \textbf{Phase 2: Client $\mathcal{C}_{r_{\max}^l-1}^l$}
\STATE $\mathcal{C}_{r_{\max}^l-1}^l$ transmits $W_{1}^{l,r_{\max}^l-1} = \mathcal{I}_{F_1}^{l,r_{\max}^l-1} \oplus \mathcal{I}_{L_1}^{l,r_{\max}^l-1}$

\FOR{$j = 1$ \TO $(C^l - r_{\max}^l - T^l)$}
    \STATE $\mathcal{C}_{r_{\max}^l-1}^l$ transmits $W_{j+1}^{l,r_{\max}^l-1} = \mathcal{I}_{F_1}^{l,r_{\max}^l-1} \oplus \mathcal{I}_{U_j}^{l,r_{\max}^l-1}$
    \ENDFOR

\FOR{$j = (C^l - r_{\max}^l -T^l + 1)$ \TO $r_{\max}^l-3$}
    \STATE $q = j - (C^l - r_{\max}^l -T^l )$

    \STATE $\mathcal{C}_{r_{\max}^l-1}^l$ transmits $W_{j+1}^{l,r_{\max}^l-1} = \mathcal{I}_{L_1}^{l,r_{\max}^l-1} \oplus \mathcal{I}_{L_{q+1}}^{l,r_{\max}^l-1}$
  \ENDFOR

\STATE \textbf{Phase 3: Client $\mathcal{C}_{r_{\max}^l}^l$}
\FOR{$j = 1$ \TO $(C^l - r_{\max}^l - T^l + 2)$}
    \STATE $\mathcal{C}_{r_{\max}^l}^l$ transmits $W_{j}^{l,r_{\max}^l} = \mathcal{I}_{F_1}^{l,r_{\max}^l} \oplus \mathcal{I}_{U_j}^{l,r_{\max}^l}$
    
\ENDFOR
\STATE $\mathcal{W}^l \gets \text{all transmissions done at this level $l$}$
\RETURN $\mathcal{W}^l $
\end{algorithmic}
\end{algorithm}
\begin{algorithm}
\caption{ Special Case Transmission Scheme}
\label{alg:special-case}
\begin{algorithmic}[1]
\REQUIRE $C^l$ clients, $\mathcal{C}^l =\{\mathcal{C}_i^l, i \in [C^l]\}$,  cardinality of side-information set of the first client in  $\mathcal{C}^l$ as $K^l$,  parameter $r_{\max}^l$ such that $C^l = \frac{(r_{\max}^l-2)(r_{\max}^l-1)}{2} + 1$
\ENSURE Transmission scheme $\mathcal{W}^l$ for level $l$

\STATE Initialize: $\mathcal{W}^l = \emptyset$

\STATE \textbf{Phase 1: Intermediate Clients ($\mathcal{C}_2^l$ to $\mathcal{C}_{r_{\max}^l-2}^l$)}
\STATE $\mathcal{C}_2^l$ transmits $W_1^{l,2} = \mathcal{I}_{F_1}^{l,2} \oplus \mathcal{I}_{F_2}^{l,2} \oplus \mathcal{I}_{L_1}^{l,2}$

\FOR{$i = 1$ \TO $r_{\max}^l-4$}
    \STATE $\mathcal{C}_{i+2}^l$ transmits $W_{1}^{l,i+2} = \mathcal{I}_{F_1}^{l,i+2} \oplus \mathcal{I}_{L_1}^{l,i+2}$
    
    \FOR{$j = 1$ \TO $i-1$}
        \STATE $\mathcal{C}_{i+2}^l$ transmits $W_{j+1}^{l,i+2} = \mathcal{I}_{F_1}^{l,i+2} \oplus \mathcal{I}_{U_j}^{l,i+2}$
           \ENDFOR
\ENDFOR

\STATE $T^l = \frac{(r_{\max}^l-3)(r_{\max}^l-4)}{2} + 1$

\STATE \textbf{{\bf Phase 2}: Client $\mathcal{C}_{r_{\max}^l-1}^l$}
\STATE $\mathcal{C}_{r_{\max}^l-1}^l$ transmits $W_{1}^{l,{r_{\max}^l-1}} = \mathcal{I}_{F_1}^{l,r_{\max}^l-1} \oplus \mathcal{I}_{L_1}^{l,r_{\max}^l-1}$

\FOR{$j = 1$ \TO $r_{\max}^l-4$}
    \STATE $\mathcal{C}_{r_{\max}^l-1}^l$ transmits $W_{j+1}^{l,{r_{\max}^l-1}} = \mathcal{I}_{L_1}^{l,r_{\max}^l-1} \oplus \mathcal{I}_{U_{j}}^{l,r_{\max}^l-1}$
   \ENDFOR

\STATE \textbf{Phase 3: Client $\mathcal{C}_{r_{\max}^l}^l$}
\FOR{$j = 1$ \TO $(C^l - r_{\max}^l - T^l + 3)$}
    \STATE $\mathcal{C}_{r_{\max}^l}^l$ transmits $W_{j}^{l,r_{\max}^l} = \mathcal{I}_{F_1}^{l,r_{\max}^l} \oplus \mathcal{I}_{U_j}^{l,r_{\max}^l}$
    
\ENDFOR

\STATE $\mathcal{W}^l \gets \text{all transmissions done at this level $l$}$
\RETURN $\mathcal{W}^l $
\end{algorithmic}
\end{algorithm}
\begin{proof}
    
We prove achievability by constructing an explicit transmission scheme and showing that it satisfies all clients using \(N(C)\) transmissions, ensuring that by the end of the process each client holds exactly \(T\) messages. Algorithm~1 solves the secure DICP for the LPS-FO side-information setting, guaranteeing that every client reaches exactly \(T = K + C\) messages through a recursive approach that systematically decomposes the original problem into smaller subproblems. All necessary notation is provided in Table~\ref{tab:notations}. The algorithm begins by initializing the recursive level \(l = 1\) and setting \(C^1 = C\), the total number of clients, and appropriately renaming the clients for consistency at each level. At each recursive level \(l\), the algorithm determines a key parameter \(r_{\max}^l\), defined as the positive integer satisfying the combinatorial inequality \eqref{rmax} based on the current number of clients \(C^l\). This parameter governs the structure of transmissions at that level. The triangular number condition in \eqref{rmax} ensures that the \(C^l\) clients are partitioned into a transmitting group of size \(r_{\max}^l\) and a residual group of size \(C^l - r_{\max}^l\), which moves to level \(l+1\).

Depending on whether \(C^l\) corresponds to a special  case or not, Algorithm 1 branches into one of two transmission subroutines: the general case (Algorithm 2) or the special case (Algorithm 3). Both subroutines generate a set of \(C^l\) coded transmissions \(\mathcal{W}^l\) from a subset of the clients in \(\mathcal{C}^l\) at that level. The transmissions are carefully designed as XOR combinations of messages from the clients' side-information sets, leveraging the initial overlapping segment, terminal overlapping segment, and unique non-overlapping messages. After transmissions for level \(l\) are generated and added to the overall scheme \(\mathcal{W}\), the algorithm prepares for the next recursive level. It is important to note that the LPS-FO structure is preserved in the side-information sets at the beginning of each recursive level for the set of clients considered at that level, which ensures the consistency of the recursive decomposition.

Upon completion of level \(l\), clients \(\mathcal{C}_1^l\) through \(\mathcal{C}_{r_{\max}^l}^l\) have attained the target knowledge level \(T\), while the remaining clients \(\mathcal{C}_{r_{\max}^l+1}^l\) through \(\mathcal{C}_{C^l}^l\) retain their original knowledge levels, having decoded no messages from the transmissions at that level. These remaining clients are re-indexed to form \(\mathcal{C}^{l+1} = \{\mathcal{C}_i^{l+1} = \mathcal{C}_{r_{\max}^l + i}^l \mid i \in [1, C^l - r_{\max}^l]\}\) for the next recursive level \(l+1\). In essence, at each level \(l\), only the first \(r_{\max}^l\) clients participate in decoding; the transmissions are designed solely to satisfy their demands, and their requirements are fully met before the algorithm proceeds.
The process repeats recursively until the number of remaining clients \(C^l\) $\leq 2$, at which point base case transmission schemes are applied. For two clients, three specific XOR transmissions ensure both can decode their remianing needed number of messages, while for a single client, one transmission from a designated original client suffices. The algorithm finally returns the complete set \(\mathcal{W}\) accumulated from all recursive levels, ensuring that all clients can decode their desired number of messages using their side-information and the received coded packets.

\end{proof}
\begin{remark}
    In the absence of security constraints, a lower bound on the number of transmissions is given by the number of clients $C$ as in Theorem \ref{thm optimality}. Under the  security constraint, the number of transmissions required for our scheme proposed in Theorem 2 exceeds the lower bound by $ N(C - r_{\max})$. This deviation indicates the cost of enforcing security.
\end{remark}

\begin{remark}
    It is observed from Algorithm 1 that if ${C}  \in \{3,4\}$, then $r_{\max} = C$. Consequently, the recursion terminates after a single level and the total number of transmissions required is $C$ which matches the lower bound in Theorem \ref{thm optimality}. Hence the proposed scheme is  optimal for these cases.
\end{remark}

\noindent Now, we illustrate  Theorem 2 with the help of an example.
\begin{example}
In this example, we illustrate the transmission scheme for a case where $M = 75$, $K = 7$, and $P = 3$ (which is considered in Example 1). Here, the number of clients is $C=9$, and  $\mathcal{C} =\{\mathcal{C}_i: i \in [9]\}$ is the set of clients present. The goal is for each client to possess exactly \(T = K + C = 16\) messages by the end of transmissions.  The side-information sets of the clients are given in Table~\ref{tab:sec_example1} and the messages each client decodes after every transmission are  given in Table~\ref{tab:sec_example2}. The transmission scheme proceeds as follows.

\begin{table*}[!t]
\centering
\caption{Decoding procedure for example 1}
\label{tab:sec_example2}
\resizebox{\linewidth}{!}{
\begin{tabular}{|c|c|c|c|c|c|c|c|c|c|c|c|c|c|}
\hline
        Client & $x_5 \oplus x_{10}$ & $x_{10} \oplus x_{16}$ & $x_{10} \oplus x_{13}$ & $x_{16} \oplus x_{23}$ & $x_{16} \oplus x_{19}$ & $x_{23} \oplus x_{24}$  & $x_{23} \oplus x_{26}$ & $x_{23} \oplus x_{27}$ & $x_{23} \oplus x_{28}$ & $x_{40} \oplus x_{41} \oplus x_{50}$ & $x_{50} \oplus x_{61}$ & $x_{61} \oplus x_{64}$ & $x_{61} \oplus x_{65}$ \\
\hline

$\mathcal{C}_1$ & $x_{10}$ & $x_{16}$ & $x_{13}$ & $x_{23}$ & $x_{19}$ & $x_{24}$ & $x_{26}$ & $x_{27}$ & $x_{28}$ & - & - & -& -\\

$\mathcal{C}_2$ & -- & $x_{16}$ & $x_{13}$ & $x_{23}$ & $x_{19}$  & $x_{24}$ & $x_{26}$ & $x_{27}$ & $x_{28}$ & - & - & -& -\\

$\mathcal{C}_3$ & $x_{5}$ & -- & -- & $x_{23}$ & $x_{19}$ & $x_{24}$ & $x_{26}$ & $x_{27}$ & $x_{28}$ & - & - & -& -\\

$\mathcal{C}_4$ & $x_{5}$ & $x_{10}$ & $x_{13}$ & -- & --  & -- & $x_{26}$ & $x_{27}$ & $x_{28}$ & - & - & -& -\\

$\mathcal{C}_5$ & $x_{5}$ & $x_{10}$ & $x_{13}$ & $x_{16}$ & $x_{19}$  & -- & -- & -- & -- & -- & - & -& -\\

$\mathcal{C}_6$ & --& -- & -- & -- & --  & -- & -- & -- & -- & $x_{50}$ & $x_{61}$ & $x_{64}$ & $x_{65}$\\

$\mathcal{C}_7$ & -- & -- & -- & -- & -- & -- & -- & -- & -- & -- & $x_{61}$ & $x_{64}$ & $x_{65}$\\

$\mathcal{C}_8$ &-- & -- & -- & -- & --  & -- & -- & -- & -- & -- & -- & $x_{64}$ & $x_{65}$\\

$\mathcal{C}_9$ & -- & -- & -- & -- & -- & -- & -- & -- & -- & -- & $x_{50}$ & -- & --\\

\hline
\end{tabular}}
\end{table*}

 First $r_{\max}^l$ for level $l=1$ is calculated. Here, $l = 1$, $C^1 = 9, K^{1} = 7$ and $P = 3$. We set $\mathcal{C}^1 =\{\mathcal{C}_i, i \in [9]\}$. From {\bf{Step 1}} of {\bf{Algorithm 1}}, $r_{\max}^1 = 5$. Hence, the first five clients will be considered in this level. Here $C^1 = 9$ and $\frac{(r_{\max}^1-2)(r_{\max}^1-1)}{2} + 1 = 7$, so $C^1 \neq \frac{(r_{\max}^1-2)(r_{\max}^1-1)}{2} + 1$. Therefore, from {{line 5}} of {\bf{Algorithm 1}}, {\bf{Algorithm 2}} is executed.
 
\textbf{Phase 1:} In this phase, clients $\mathcal{C}_2$ and $\mathcal{C}_3$ perform transmissions. First, $\mathcal{C}_2$ sends the coded message $W_{1}^{1,2} = x_5 \oplus x_{10}$, which consists of the first message from its overlap with $\mathcal{C}_1$ and the first from its overlap with $\mathcal{C}_3$. This allows $\mathcal{C}_1$ to decode $x_{10}$ (since $\mathcal{C}_1$ has already decoded $x_{10}$) and $\mathcal{C}_3$ to decode $x_5$. Next, $\mathcal{C}_3$ transmits $W_{1}^{1,3} = x_{10} \oplus x_{16}$, combining the first message from its overlap with $\mathcal{C}_2$ and the first from its overlap with $\mathcal{C}_4$. From this, $\mathcal{C}_1$ and $\mathcal{C}_2$ decode $x_{16}$, while $\mathcal{C}_4$ decodes $x_{10}$; having obtained $x_{10}$, $\mathcal{C}_4$ can now also decode $x_5$ from $W_{1}^{1,2}$. Finally, $\mathcal{C}_3$ sends $W_{2}^{1,3} = x_{10} \oplus x_{13}$, formed from the first message of its overlap with $\mathcal{C}_2$ and its first unique message, enabling $\mathcal{C}_1$, $\mathcal{C}_2$, and $\mathcal{C}_4$ to each decode $x_{13}$. According to line 9 of \textbf{Algorithm 2}, the total number of transmissions in this phase $T^1 = 3$.

\textbf{Phase 2:} Here, client $\mathcal{C}_4$ transmits. It begins with $W_{1}^{1,4} = x_{16} \oplus x_{23}$, which encodes the first message from its overlap with $\mathcal{C}_3$ and the first from its overlap with $\mathcal{C}_5$. This allows $\mathcal{C}_1$, $\mathcal{C}_2$, and $\mathcal{C}_3$ to decode $x_{23}$, while $\mathcal{C}_5$ decodes $x_{16}$. With $x_{16}$ now known, $\mathcal{C}_5$ can sequentially decode $x_{10}$, $x_5$, and $x_{13}$ from the earlier transmissions $W_1^{1,3}$, $W_1^{1,2}$, and $W_2^{1,3}$. Next, $\mathcal{C}_4$ sends $W_{2}^{1,4} = x_{16} \oplus x_{19}$, combining its overlap message with $\mathcal{C}_3$ and its first unique message, so that $\mathcal{C}_1$, $\mathcal{C}_2$, $\mathcal{C}_3$, and $\mathcal{C}_5$ all decode $x_{19}$. At this point, client $\mathcal{C}_5$ has reached its target level because $\mathcal{C}_5$'s side-information set already contains 11 messages, and it required only 5 more to reach its target and should therefore receive no further messages. Consequently, $\mathcal{C}_4$ transmits $W_{3}^{1,4} = x_{23} \oplus x_{24}$, formed from the first and second messages of its overlap with $\mathcal{C}_5$. This transmission provides no new messages to $\mathcal{C}_5$, but enables $\mathcal{C}_1$, $\mathcal{C}_2$, and $\mathcal{C}_3$ to decode $x_{24}$. After this phase, the first four clients each need exactly three more new messages to reach their target level.

\textbf{Phase 3:} In the final phase, client $\mathcal{C}_5$ transmits. It sends $W_{1}^{1,5} = x_{23} \oplus x_{26}$, encoding the first message from its overlap with $\mathcal{C}_4$ and its first unique message; $\mathcal{C}_1$, $\mathcal{C}_2$, $\mathcal{C}_3$, and $\mathcal{C}_4$ all decode $x_{26}$ from this. Next, $\mathcal{C}_5$ transmits $W_{2}^{1,5} = x_{23} \oplus x_{27}$, similarly formed from the same overlap message and its second unique message, allowing the same four clients to decode $x_{27}$. Finally, $\mathcal{C}_5$ sends $W_{3}^{1,5} = x_{23} \oplus x_{28}$, combining the overlap message with its third unique message, which provides $x_{28}$ to $\mathcal{C}_1$, $\mathcal{C}_2$, $\mathcal{C}_3$, and $\mathcal{C}_4$. 

With these transmissions, the first five clients have all reached their target message counts. Next, we move to the next recursive level. Here, $l = 2$, $C^2 = 4,K^{2} = 12$ and $\mathcal{C}^2 =\{\mathcal{C}_i, i \in [6,9]\}$. From {\bf{Step 1}} of {\bf{Algorithm 1}}, $r_{\max}^2 = 4$. Hence, the last four clients will be considered in this level. Here $C^2 = 4$ and $\frac{(r_{\max}^2-2)(r_{\max}^2-1)}{2} + 1 = 4$, so $C^2 = \frac{(r_{\max}^1-2)(r_{\max}^1-1)}{2} + 1$. Therefore, {\bf{Algorithm 3}} is executed.

\textbf{Phase 1:} In this phase, client $\mathcal{C}_7$ transmits a coded message composed of the first and second messages from its overlap with $\mathcal{C}_6$ and the first message from its overlap with $\mathcal{C}_8$, specifically $W_{1}^{2,2} = x_{40} \oplus x_{41} \oplus x_{50}$. This construction ensures that only $\mathcal{C}_6$ can decode a new message ($x_{50}$) from this transmission.

\textbf{Phase 2:} Here, client $\mathcal{C}_8$ transmits the coded message $W_{1}^{2,3} = x_{50} \oplus x_{61}$, formed from the first message in its overlap with $\mathcal{C}_7$ and the first message in its overlap with $\mathcal{C}_9$. This allows $\mathcal{C}_6$ and $\mathcal{C}_7$ to decode $x_{61}$, while $\mathcal{C}_9$ decodes $x_{50}$. After this transmission, client $\mathcal{C}_9$ has reached its target, leaving clients $\mathcal{C}_6$, $\mathcal{C}_7$, and $\mathcal{C}_8$ each requiring two more new messages.

\textbf{Phase 3:} In the final phase, client $\mathcal{C}_9$ makes two transmissions. First, it transmits $W_{1}^{2,4} = x_{61} \oplus x_{64}$, combining the first message from its overlap with $\mathcal{C}_8$ and its first unique message, enabling $\mathcal{C}_6$, $\mathcal{C}_7$, and $\mathcal{C}_8$ to decode $x_{64}$. Second, it transmits $W_{2}^{2,4} = x_{61} \oplus x_{65}$, formed similarly from the same overlap message and its second unique message, which provides $x_{65}$ to the same three clients. 

With these transmissions, all  clients have reached their exact target message counts.  

\end{example}

\section{Proof of Correctness of Algorithm 1}
\label{sec:proof-alg1}

This section establishes the correctness of Algorithm 1, the main recursive algorithm. The algorithm operates across iterative levels. At each level \(l\), it considers a client set \(\mathcal{C}^l\) of size \(C^l\), where the side-information set of the first client \(\mathcal{C}_1^l\) has cardinality \(K^l\). At every level, the algorithm determines a parameter \(r_{\max}^l\) using inequality (\ref{rmax}). This parameter specifies how many clients from the current set can be fully served at that level.

Based on the current client count \(C^l\), Algorithm 1 selects an appropriate transmission scheme: Algorithm \ref{alg:general-case} for the general case, or Algorithm \ref{alg:special-case} when a specific equality condition holds. These schemes are designed so that transmissions at level \(l\) satisfy only the first \(r_{\max}^l\) clients in \(\mathcal{C}^l\), leaving the remaining clients in \(\mathcal{C}^l\) entirely unaffected. This property is formally guaranteed by Lemmas \ref{lemma:alg2} and \ref{lemma:alg3}, proven in Sections \ref{sec:proof-alg2} and \ref{sec:proof-alg3}, respectively.

Furthermore, Lemma \ref{lemma:general} ensures that clients outside the current recursive set \(\mathcal{C}^l\) cannot decode any messages transmitted at level \(l\). This containment is crucial for the recursive structure.

\begin{lemma}
\label{lemma:alg2}
    At each level \(l\), if \(C^l \neq \frac{(r_{\max}^l-2)(r_{\max}^l-1)}{2} + 1\), then Algorithm \ref{alg:general-case} provides \(C^l\) transmissions that enable the first \(r_{\max}^l\) clients in \(\mathcal{C}^l\) (each client \(\mathcal{C}_i^l\) with \(i \leq r_{\max}^l\)) to achieve the target size \(T=K+C\). No other client in \(\mathcal{C}^l\) except these first \(r_{\max}^l\) clients can decode any message transmitted during level \(l\).
\end{lemma}

\begin{lemma}
\label{lemma:alg3}
    At each level \(l\), if \(C^l = \frac{(r_{\max}^l-2)(r_{\max}^l-1)}{2} + 1\), then Algorithm \ref{alg:special-case} provides \(C^l\) transmissions that enable the first \(r_{\max}^l\) clients in \(\mathcal{C}^l\) (each client \(\mathcal{C}_i^l\) with \(i \leq r_{\max}^l\)) to achieve the target size \(T=K+C\). No other client in \(\mathcal{C}^l\) except these first \(r_{\max}^l\) clients can decode any message transmitted during level \(l\).
\end{lemma}

\begin{lemma}
\label{lemma:general}
    At each level \(l\), no client in \(\mathcal{C} \backslash \mathcal{C}^l\) can decode any message transmitted during this level.
\end{lemma}

At any recursive level \(l\), after computing \(r_{\max}^l\) and executing the chosen transmission scheme (Algorithm 2 or 3), the first \(r_{\max}^l\) clients achieve the target knowledge level \(T = K + C\). By the aforementioned lemmas, no other client in \(\mathcal{C}^l\) gains any information from these transmissions. Consequently, these satisfied clients are removed from subsequent steps. The remaining \(C^l - r_{\max}^l\) clients, with their original side-information sets intact, are re-indexed to form the client set for the next level:
\[
\mathcal{C}^{l+1} = \{\mathcal{C}_i^{l+1} = \mathcal{C}_{r_{\max}^l + i}^l : i \in [1, C^l - r_{\max}^l]\}.
\]
The parameters are updated as \(C^{l+1} = C^l - r_{\max}^l\) and \(K^{l+1} = K^l + r_{\max}^l\).

At any level $l$, the transmissions are constructed in such that, for any client outside $C_l$, neither of the two messages is any transmission is available in its side-information set. Hence, no such client can decode any transmission.

\begin{lemma}
\label{lemma side info structure}
At the beginning of each recursive level \(l\), the side-information sets maintain the LPS-FO structure with respect to the \(C^l\) clients in \(\mathcal{C}^l= \{\mathcal{C}_1^l, \mathcal{C}_2^l, \ldots, \mathcal{C}_{C^l}^l\}\), where the first client \(\mathcal{C}_1^l\) has a side-information set of cardinality \(K^l\).
\end{lemma}

The proof of Lemma \ref{lemma side info structure} is provided in Section \ref{sec: recursive-structure}. This recursive process continues until the number of remaining clients is reduced to the base cases \(C^l = 2\) or \(C^l = 1\), which are handled explicitly in Algorithm 1.

\begin{lemma}
\label{base case lemma}
    For the base cases in Algorithm 1, \(N(2)=3\) (three transmissions are required when two clients remain) and \(N(1)=1\) (one transmission is required when one client remains).
\end{lemma}

The proof of Lemma \ref{base case lemma} is provided in Section \ref{sec:base-cases}. Therefore, Algorithm 1 produces a complete transmission scheme \(\mathcal{W}\) such that every client in the system attains the target knowledge level \(T = K + C\).

To calculate the total number of transmissions \(N(C)\), we note that at level 1 there are \(C\) transmissions, satisfying the first \(r_{\max}\) clients, where \(r_{\max}\) is determined by (\ref{rmax general}) as per Algorithm 1 (which is equivalent to (\ref{rmax})). The remaining \(C - r_{\max}\) clients enter level 2, where the problem is treated anew with \(C - r_{\max}\) clients. This yields the recurrence:
\[
N(C) = C + N(C - r_{\max}).
\]
The recursion continues until \(C^l \leq 2\), with base cases requiring \(N(2) = 3\), \(N(1) = 1\), and \(N(0) = 0\) transmissions. Since \(r_{\max}^l \geq 2\) for all \(C^l > 2\), the sequence \(C^{l+1} = C^l - r_{\max}^l\) is strictly decreasing, guaranteeing that the algorithm terminates in finite steps. Upon termination, every client possesses exactly \(T\) messages, with no client exceeding this target, thereby satisfying both the completeness and security requirements of the DPIC problem under the LPS-FO model.

This completes the proof of correctness for Algorithm~1.

\section{Proof of Correctness of Algorithm 2}
\label{sec:proof-alg2}

This section provides a comprehensive proof of correctness for Algorithm~2, which handles the general case in our recursive transmission scheme. The proof is structured according to the three phases of the algorithm, demonstrating that after completion of all phases, each client $\mathcal{C}_i^l \in \mathcal{C}^l$ with $i \leq r_{\max}^l$ attains exactly $T = K + C$ messages, while clients with $i > r_{\max}^l$ remain unaffected, preserving the recursive structure.

\subsection{{{\bf Phase 1}}: Intermediate Clients Transmission Structure}
\label{subsec:phase1-proof-alg2}

For each intermediate client $\mathcal{C}_i^l$, where $i \in [2, r_{\max}^l-2]$, the transmission pattern consists of exactly $i-1$ coded messages. The first transmission takes the form $
W_{1}^{l,i+1} = \mathcal{I}_{F_1}^{l,i+1} \oplus \mathcal{I}_{L_1}^{l,i+1},
$
which combines one message from the initial overlapping segment (the first $P$ messages) with one message from the terminal overlapping segment (the last $P$ messages) of the client's side-information set. The remaining $i-2$ transmissions are constructed as
$
W_{j+1}^{l,i+1} = \mathcal{I}_{F_1}^{l,i+1} \oplus \mathcal{I}_{U_j}^{l,i+1}, \quad j \in [1, i-2],
$
where each XOR involves the same initial overlapping message $\mathcal{I}_{F_1}^{l,i+1}$ paired with a distinct unique (non-overlapping) message $\mathcal{I}_{U_j}^{l,i+1}$ from the client's side-information set.
The validity of these transmissions rests on the assumption $K \geq 2P$, which guarantees that every client $\mathcal{C}_i^l$ possesses at least $i-1$ unique messages. This condition follows from the LPS-FO structure.  Consequently, all required XOR combinations are well-defined and can be generated by the transmitting client.

The transmission pattern enables a dual decoding mechanism where each client decodes exactly one new message from every transmission made by subsequent clients and every transmission made by preceding clients.

Consider client $\mathcal{C}_i^l$ with $1 \leq i \leq r_{\max}^l-3$. For decoding from subsequent clients, the process begins with client $\mathcal{C}_{i+1}^l$, from whose transmissions $\mathcal{C}_i^l$ can decode messages. Since $\mathcal{C}_i^l$ knows $\mathcal{I}_{F_1}^{l,i+1}$ through overlap with $\mathcal{C}_{i+1}^l$ ($\mathcal{I}_{F_1}^{l,i+1}=\mathcal{I}_{L_1}^{l,i}$ by Definition \ref{def LPSFO}), it can decode $\mathcal{I}_{L_1}^{l,i+1}$ from the transmission $W_{1}^{l,i+1} = \mathcal{I}_{F_1}^{l,i+1} \oplus \mathcal{I}_{L_1}^{l,i+1}$. Additionally, from transmissions of the form $W_{j+1}^{l,i+1} = \mathcal{I}_{F_1}^{l,i+1} \oplus \mathcal{I}_{U_j}^{l,i+1}$, it decodes each $\mathcal{I}_{U_j}^{l,i+1}$ for $j \in [1, i-2]$. A forward decoding chain  occurs then: by the LPS-FO structure, $\mathcal{I}_{F_1}^{l,i+2} = \mathcal{I}_{L_1}^{l,i+1}$, which $\mathcal{C}_i^l$ now possesses. This enables $\mathcal{C}_i^l$ to decode from $\mathcal{C}_{i+2}^l$'s transmissions, obtaining $\mathcal{I}_{L_1}^{l,i+2}$ and $\mathcal{I}_{U_j}^{l,i+2}$. This process continues recursively for all clients $k$ with $i < k \leq r_{\max}^l-2$, allowing $\mathcal{C}_i^l$ to accumulate the set $\{\mathcal{I}_{L_1}^{l,k} \mid k \in [i+1, r_{\max}^l-2]\} \cup \{\mathcal{I}_{U_j}^{l,k} \mid k \in [i+1, r_{\max}^l-2], j \in [1, k-2]\}$ from subsequent clients.

Simultaneously, $\mathcal{C}_i^l$, for $i \in [3, r_{\max}-1]$ engages in backward decoding from preceding clients. Client $\mathcal{C}_i^l$ knows $\mathcal{I}_{L_1}^{l,i-1}$ (which overlaps with its own side-information set, i.e., $\mathcal{I}_{F_1}^{l,i}=\mathcal{I}_{L_1}^{l,i-1}$). From the transmission $W_{1}^{l,i-1} = \mathcal{I}_{F_1}^{l,i-1} \oplus \mathcal{I}_{L_1}^{l,i-1}$ by client $\mathcal{C}_{i-1}^l$, it receives $\mathcal{I}_{F_1}^{l,i-1}$, enabling it to further decode $\mathcal{I}_{U_j}^{l,i-1}$ from transmissions $W_{j+1}^{l,i-1} =  \mathcal{I}_{F_1}^{l,i-1} \oplus \mathcal{I}_{U_j}^{l,i-1}$, for $j \in [i-3]$. The backward decoding chain continues  because $\mathcal{I}_{F_1}^{l,i-1} = \mathcal{I}_{L_1}^{l,i-2}$, allowing $\mathcal{C}_i^l$ to decode from $\mathcal{C}_{i-2}^l$, obtaining $\mathcal{I}_{F_1}^{l,i-2}$ and subsequently $\mathcal{I}_{U_j}^{l,i-2}$. This iterative backward decoding extends to all clients $k$ with $2 \leq k < i$, providing $\mathcal{C}_i^l$ with the set $\{\mathcal{I}_{F_1}^{l,k} \mid k \in [2, i-1]\} \cup \{\mathcal{I}_{U_j}^{l,k} \mid k \in [2, i-1], j \in [1, k-2]\}$.

Each intermediate client $\mathcal{C}_i^l$, for $i \in [2, r_{\max}^l-2]$, transmits exactly $i-1$ messages during \textbf{Phase~1}, yielding a total of $T^l = \sum_{i=2}^{r_{\max}^l-2} (i-1) = \frac{(r_{\max}^l-2)(r_{\max}^l-3)}{2}$ transmissions for this phase. Let $T_i^{l,1}$ denote the number of new messages decoded by client $\mathcal{C}_i^l$ (where $i \in [1, r_{\max}^l-1]$) from \textbf{Phase~1} transmissions. Since each transmission contains exactly one unknown message that the appropriate client can recover, every client obtains one new message from every transmission it is able to decode. For a transmitting client $\mathcal{C}_i^l$ with $i \in [2, r_{\max}^l-2]$, the $i-1$ transmissions it makes are not decodable by itself; consequently, it can decode from the remaining $T^l - (i-1)$ transmissions. Hence, $T_i^{l,1} = T^l - (i-1)$ for $i \in [2, r_{\max}^l-2]$. The two non‑transmitting clients in this phase, $\mathcal{C}_1^l$ and $\mathcal{C}_{r_{\max}^l-1}^l$, participate fully in the forward and backward decoding chains, respectively, and therefore each decodes exactly one new message from every \textbf{Phase~1} transmission. Thus, $T_1^{l,1} = T_{r_{\max}^l-1}^{l,1} = T^l$. After \textbf{Phase~1}, each client $\mathcal{C}_i^l$ with $i \in [1, r_{\max}^l-1]$ has acquired $T_i^{l,1}$ additional messages beyond its initial side‑information set.

\subsection{{\bf Phase 2}: Client $\mathcal{C}_{r_{\max}^l-1}^l$ Transmission Strategy}
\label{subsec:phase2-proof-alg2}

{\bf Phase 2} consists exclusively of transmissions from client $\mathcal{C}_{r_{\max}^l-1}^l$, which transmits exactly $r_{\max}^l - 2$ coded messages. These transmissions are organized into two complementary categories that serve distinct purposes for different client groups.

The first category, comprising $C^l - r_{\max}^l - T^l+1$ transmissions, includes
$W_{1}^{l,r_{\max}^l-1} = \mathcal{I}_{F_1}^{l,r_{\max}^l-1} \oplus \mathcal{I}_{L_1}^{l,r_{\max}^l-1}$
and
$W_{j+1}^{l,r_{\max}^l-1} = \mathcal{I}_{F_1}^{l,r_{\max}^l-1} \oplus \mathcal{I}_{U_j}^{l,r_{\max}^l-1}$
for $j \in [1, C^l - r_{\max}^l - T^l]$.
We first verify the validity of these transmissions. The quantity
$C^l - r_{\max}^l - T^l$
is non-negative because from inequality (\ref{rmax}) we have
$C^l > \frac{(r_{\max}^l-2)(r_{\max}^l-1)}{2}$,
while
$T^l = \frac{(r_{\max}^l-2)(r_{\max}^l-3)}{2}$.
A direct comparison gives $C^l - T^l > r_{\max}^l - 2$, which implies
$C^l - r_{\max}^l - T^l > 0$.
From inequality (\ref{rmax}) we have 
$C^l \le \frac{(r_{\max}^l-1)r_{\max}^l}{2}$.
Combining this with
$T^l = \frac{(r_{\max}^l-3)(r_{\max}^l-2)}{2}$
yields
\begin{align}
C^l - T^l
&\le \frac{(r_{\max}^l-1)r_{\max}^l}{2}
   - \frac{(r_{\max}^l-2)(r_{\max}^l-3)}{2} \nonumber \\
&\le 2r_{\max}^l - 3.
\end{align}
Consequently,
\begin{align}
C^l - r_{\max}^l - T^l
\le (2r_{\max}^l - 3) - r_{\max}^l
= r_{\max}^l - 3.
\label{unique}
\end{align}
Thus, the number of required unique messages,
$C^l - r_{\max}^l - T^l$,
is at most $r_{\max}^l - 3$, which is strictly less than $r_{\max}^l - 2$.

Under the assumption $K \ge 2P$, each client $\mathcal{C}_i^l$ is guaranteed to have at least $i-1$ unique (non-overlapping) messages within its side-information set. Therefore, client $\mathcal{C}_{r_{\max}^l-1}^l$ possesses at least $r_{\max}^l - 2$ unique messages. The transmissions $W_{j+1}^{l,r_{\max}^l-1}$ demand exactly $C^l - r_{\max}^l - T^l$ distinct unique messages. Since
$C^l - r_{\max}^l - T^l < r_{\max}^l - 2$ holds (from (\ref{unique})), the required number of unique messages does not exceed what is available. Consequently, all XOR combinations prescribed for this phase are well-defined and can be generated by the transmitting client.

These transmissions help to continue forward decoding chains for clients in the set
$\{\mathcal{C}_i^l, i \in [r_{\max}^l-2]\}$
and introduce backward decoding chain for the client $\mathcal{C}_{r_{\max}^l}^l$.
For any client $\mathcal{C}_i^l$ with $i < r_{\max}^l-1$, knowledge of
$\mathcal{I}_{F_1}^{l,r_{\max}^l-1}$
(obtained from {\bf Phase 1} since $\mathcal{I}_{F_1}^{l,r_{\max}^l-1}=\mathcal{I}_{L_1}^{l,r_{\max}^l-2}$) enables decoding of
$\mathcal{I}_{L_1}^{l,r_{\max}^l-1}$
from $W_{1}^{l,r_{\max}^l-1}$ and
$\mathcal{I}_{U_j}^{l,r_{\max}^l-1}$
from $W_{j+1}^{l,r_{\max}^l-1}$
for $j \in [1, C^l - r_{\max}^l - T^l]$.
Client $\mathcal{C}_{r_{\max}^l}^l$ can decode
$\mathcal{I}_{F_1}^{l,r_{\max}^l-1}$
from $W_{1}^{l,r_{\max}^l-1}$ as
$\mathcal{I}_{L_1}^{l,r_{\max}^l-1}
= \mathcal{I}_{F_1}^{l,r_{\max}^l}$ is available with it.
So, it can further decode
$\mathcal{I}_{U_j}^{l,r_{\max}^l-1}$
from $W_{j+1}^{l,r_{\max}^l-1}$
for $j \in [1, C^l - r_{\max}^l - T^l]$,
because it knows
$\mathcal{I}_{F_1}^{l,r_{\max}^l-1}$.

The second category consists of
$r_{\max}^l - 3 - (C^l - r_{\max}^l - T^l)$
transmissions for
$j \in [C^l - r_{\max}^l - T^l + 1, r_{\max}^l - 3]$
of the form
$W_{j+1}^{l,r_{\max}^l-1}
= \mathcal{I}_{L_1}^{l,r_{\max}^l-1}
\oplus
\mathcal{I}_{L_{q+1}}^{l,r_{\max}^l-1}$,
where
$q = j - (C^l - r_{\max}^l - T^l)$.
By assumption, we have $P \ge r_{\max} - 2$, where $r_{\max}$ is defined as the integer satisfying
\[
\frac{(r_{\max} - 2)(r_{\max} - 1)}{2}
< C \le
\frac{r_{\max}(r_{\max} - 1)}{2}.
\]
Since $C^l < C$, it follows that $r_{\max}^l \le r_{\max}$.
Consequently, $P \ge r_{\max}^l - 2$ guarantees that the terminal overlapping segment of every client contains at least $r_{\max}^l - 2$ messages.
This ensures that all required messages of the form
$\mathcal{I}_{L_j}^{l,k}$
(for $j \le r_{\max}^l - 2$)
exist within the overlapping region, thereby validating the feasibility of the corresponding XOR combinations in the transmission scheme.

These transmissions exclusively extend the forward decoding chain for clients $i < r_{\max}^l-1$. After decoding $\mathcal{I}_{L_1}^{l,r_{\max}^l-1}$ from $W_{1}^{l,r_{\max}^l-1}$, these clients can further decode $\mathcal{I}_{L_{q+1}}^{l,r_{\max}^l-1}$. For client $\mathcal{C}_{r_{\max}^l}^l$, both $\mathcal{I}_{L_1}^{l,r_{\max}^l-1}$ and $\mathcal{I}_{L_{q+1}}^{l,r_{\max}^l-1}$ are already present in its side-information set (since $\mathcal{I}_{F_j}^{l,r_{\max}^l} = \mathcal{I}_{L_j}^{l,r_{\max}^l-1}$ by the LPS-FO structure), so it gains no new message from these transmissions.

Let $T_{i}^{l,2}$ be the number of messages decoded by the client $\mathcal{C}_i^l$ for $i \in [1, r_{\max}]$ after completing {\bf Phase 2} of this level $l$.
After {\bf Phase 2}, intermediate clients $\mathcal{C}_i^l$ with $i < r_{\max}^l-1$ have decoded an additional $r_{\max}^l - 3$ messages, bringing their total decoded message count to $T_{i}^{l,2}=T_i^{l,1} +(r_{\max}^l - 3)$. Client $\mathcal{C}_{r_{\max}^l-1}^l$, being the transmitter, gains no new messages from its own transmissions, maintaining a total of $T_{r_{\max}^l-1}^{l,2}=T_{r_{\max}^l-1}^{l,1}$ decoded messages. 

Client $\mathcal{C}_{r_{\max}^l}^l$ decodes $C^l - r_{\max}^l - T^l+1 $ new messages from {\bf Phase 2} (specifically, from category 1). Now, since it has decoded the message $\mathcal{I}_{F_1}^{l,r_{\max}^l-1}$ combining with the fact that $\mathcal{I}_{F_1}^{l,r_{\max}^l-1}=\mathcal{I}_{L_1}^{l,r_{\max}^l-2}$ it enters the backward chain and decodes all the messages in the set $\{\mathcal{I}_{F_1}^{l,k} \mid k \in [2, r_{\max}^l-2]\} \cup \{\mathcal{I}_{U_j}^{l,k} \mid k \in [2, r_{\max}^l-2], j \in [1, k-2]\}$.  This results in a cumulative total of $T_{r_{\max}^l}^{l,2}=T^l+ (C^l - r_{\max}^l - T^l +1) = C^l - r_{\max}^l  + 1$ decoded messages. This brings $\mathcal{C}_{r_{\max}^l}^l$ to its target message count, as will be verified subsequently (proved later in this section).

\subsection{{\bf Phase 3}: Client $\mathcal{C}_{r_{\max}^l}^l$ Transmission Strategy}
\label{subsec:phase3-proof-alg2}

{\bf Phase 3} completes the delivery process through transmissions exclusively from client $\mathcal{C}_{r_{\max}^l}^l$, which sends exactly $C^l - r_{\max}^l - T^l + 2$ messages, each of the form
$W_j^{l,r_{\max}^l} = \mathcal{I}_{F_1}^{l,r_{\max}^l} \oplus \mathcal{I}_{U_j}^{l,r_{\max}^l}$
for $j \in [1, C^l - r_{\max}^l - T^l + 2]$.
The decodability of these transmissions relies on the receiving clients' knowledge of $\mathcal{I}_{F_1}^{l,r_{\max}^l}$.
As established earlier, $C^l - T^l > r_{\max}^l - 2$, guaranteeing that
$C^l - r_{\max}^l - T^l + 2$
is positive.

Under the assumption $K \ge 2P$, each client $\mathcal{C}_i^l$ possesses at least $i-1$ unique messages.
For $i = r_{\max}^l$, this means at least $r_{\max}^l - 1$ unique messages are available.
The inequality
$C^l - r_{\max}^l - T^l + 2 \le r_{\max}^l - 1$
follows from (\ref{unique}) (since
$C^l - r_{\max}^l - T^l \le r_{\max}^l - 3$ implies
$C^l - r_{\max}^l - T^l + 2 \le r_{\max}^l - 1$).
Therefore, the required number of unique messages for the transmissions does not exceed what is available at
$\mathcal{C}_{r_{\max}^l}^l$, ensuring all XOR combinations are well-defined and valid.

Let $T_{i}^{l,3}$ be the number of messages decoded by the client
$\mathcal{C}_i^l$ for $i \in [1, r_{\max}^l]$
after completing {\bf Phase 3} of this level $l$.
For intermediate clients $\mathcal{C}_i^l$ with
$i \in [2, r_{\max}^l-2]$,
from {\bf Phase 2}, it has already decoded
$\mathcal{I}_{F_1}^{l,r_{\max}^l}$
since
$\mathcal{I}_{F_1}^{l,r_{\max}^l}
= \mathcal{I}_{L_1}^{l,r_{\max}^l-1}$.
These clients obtained
$\mathcal{I}_{L_1}^{l,r_{\max}^l-1}$
from {\bf Phase 2}
(specifically from $W_1^{l,r_{\max}^l-1}$ of category~1),
thus they possess
$\mathcal{I}_{F_1}^{l,r_{\max}^l}$.
Consequently, from each transmission
$W_j^{l,r_{\max}^l}
= \mathcal{I}_{F_1}^{l,r_{\max}^l}
\oplus
\mathcal{I}_{U_j}^{l,r_{\max}^l}$,
they decode
$\mathcal{I}_{U_j}^{l,r_{\max}^l}$,
acquiring
$C^l - r_{\max}^l - T^l + 2$
new messages.
Hence, the total messages acquired so far in this level $l$ is
\[
T_{i}^{l,3}
= T_{i}^{l,2}
+ C^l - r_{\max}^l - T^l + 2
= C^l - (i-1).
\]

Client $\mathcal{C}_{r_{\max}^l-1}^l$ naturally knows
$\mathcal{I}_{F_1}^{l,r_{\max}^l}$
as it equals
$\mathcal{I}_{L_1}^{l,r_{\max}^l-1}$,
which is part of its own side-information set.
Therefore, it can directly decode
$\mathcal{I}_{U_j}^{l,r_{\max}^l}$
from each
$W_j^{l,r_{\max}^l}$,
also gaining
$C^l - r_{\max}^l - T^l + 2$
new messages.
Hence, the total messages acquired so far in this level $l$ is
\[
T_{r_{\max}^l-1}^{l,3}
= T_{r_{\max}^l-1}^{l,2}
+ C^l - r_{\max}^l - T^l + 2
= C^l - (r_{\max}^l - 2).
\]

Client $\mathcal{C}_{r_{\max}^l}^l$,
being the transmitter,
gains no additional messages from {\bf Phase 3} transmissions.
Hence, the total messages acquired so far in this level $l$ is
\[
T_{r_{\max}^l}^{l,3}
= T_{r_{\max}^l}^{l,2}
= C^l - (r_{\max}^l - 1).
\]

In short, by the end of Phase~3, each client
$\mathcal{C}_i^l \in \mathcal{C}^l$
has decoded
$T_{i}^{l,3}
= C^l - (i-1)$
messages from the transmissions made at this level.

Since $T = K + C$ by definition, and we have
$K^{l} = K^{l-1} + r_{\max}^{l-1}$
and
$C^{l} = C^{l-1} - r_{\max}^{l-1}$,
it follows that
$K^{l} + C^{l} = K^{l-1} + C^{l-1}$.
By induction,
$K^l + C^l = K + C = T$,
establishing that
$T = K^{l} + C^{l}$.
This confirms that the number of additional messages required by the first client in the set
$\mathcal{C}^{l}$
to reach the target remains exactly
$C^{l}$.
Hence, the number of additional messages required for any client
$\mathcal{C}^{l}_i \in \mathcal{C}^{l}$,
for $i \in [C^l]$,
to reach the target is exactly
$C^{l}- (i-1)$.
We have proved that each client
$\mathcal{C}_i^l \in \mathcal{C}^l$, for $i \leq r_{\max}^l$
has decoded
$T_{i}^{l,3}
= C^l- (i-1)$
messages from the transmissions made at this level.
Hence each of these clients 
are satisfied from the transmissions made at this level.

 Clients $\mathcal{C}_i^l \in \mathcal{C}^l$ with $i > r_{\max}^l$ receive zero messages during the current recursive level.
A crucial property emerges for clients with index $i > r_{\max}^l$. These clients cannot decode any messages from Phase~1 transmissions. All the transmissions $W$ in this level is of the type $W = \mathcal{I}_{F_1}^{l,k} \oplus \mathcal{I}_{L_1}^{l,k}$, or $W = \mathcal{I}_{L_1}^{l,k} \oplus \mathcal{I}_{L_j}^{l,k}$ (with $k \leq r_{\max}^l-1$), or   $W = \mathcal{I}_{F_1}^{l,k} \oplus \mathcal{I}_{U_j}^{l,k}$, for $k \leq r_{\max}^l$,  for some integer $j$. The  client $\mathcal{C}_i^l$ (where $i > r_{\max}^l$) lacks knowledge of both XORed terms in all these kind of transmissions as they do not have any overlap with any of the messages present in the XOR. No linear combination of received transmissions reduces the number of unknown terms to one as from the transmission structure in each of the phases, any two transmissions have atmost one message in common. Consequently, clients $\mathcal{C}_{i}^l$, for $i>r_{\max}^l$, remain unaffected by the transmissions in this level, preserving their original message sets and enabling the recursive decomposition.

\section{Proof of Correctness of Algorithm 3}
\label{sec:proof-alg3}

This section provides a comprehensive proof of correctness for Algorithm~3, which handles the special case in our recursive transmission scheme. The proof is structured according to the three phases of the algorithm, demonstrating that after completion of all phases, each client $\mathcal{C}_i^l \in \mathcal{C}^l$ with $i \leq r_{\max}^l$ attains exactly $T = K + C$ messages, while clients with $i > r_{\max}^l$ remain unaffected, preserving the recursive structure.

\subsection{{{\bf Phase 1}}: Intermediate Clients Transmission Structure}
\label{subsec:phase1-proof-alg3}

The client $\mathcal{C}_2^l$ makes a single transmission of the form $W_{1}^{l,2} = \mathcal{I}_{F_1}^{l,2} \oplus \mathcal{I}_{F_2}^{l,2} \oplus \mathcal{I}_{L_1}^{l,2}$, which combines the first two messages from the initial overlapping segment (the first $P$ messages) with one message from the terminal overlapping segment (the last $P$ messages) of the client's side-information set. This transmission is done to ensure that only the first client receives a new message.
For each intermediate client $\mathcal{C}_i^l$, where $i \in [3, r_{\max}^l-2]$, the transmission pattern consists of exactly $i-2$ coded messages. The first transmission takes the form $
W_{1}^{l,i} = \mathcal{I}_{F_1}^{l,i} \oplus \mathcal{I}_{L_1}^{l,i},
$
which combines one message from the initial overlapping segment (the first $P$ messages) with one message from the terminal overlapping segment (the last $P$ messages) of the client's side-information set. The remaining $i-3$ transmissions are constructed as
$
W_{j+1}^{l,i} = \mathcal{I}_{F_1}^{l,i} \oplus \mathcal{I}_{U_j}^{l,i}, \quad j \in [1, i-3],
$
where each XOR involves the same initial overlapping message $\mathcal{I}_{F_1}^{l,i}$ paired with a distinct unique (non-overlapping) message $\mathcal{I}_{U_j}^{l,i}$ from the client's side-information set. The validity of these transmissions rests on the assumption $K \geq 2P$, which guarantees that every client $\mathcal{C}_i^l$ possesses at least $i-1$ unique messages. This condition follows from the LPS-FO structure.  Consequently, all required XOR combinations are well-defined and can be generated by the transmitting client.

Consider client $\mathcal{C}_i^l$ with $1 \leq i \leq r_{\max}^l-3$. For decoding from subsequent clients, the process begins with client $\mathcal{C}_{i+1}^l$, from whose transmissions $\mathcal{C}_{i}^l$ can decode messages. Consider client $\mathcal{C}_1^l$ ($i=1$). Since $\mathcal{C}_1^l$ knows $\mathcal{I}_{F_1}^{l,2}$ and $\mathcal{I}_{F_2}^{l,2}$ through overlap with $\mathcal{C}_{2}^l$ ($\mathcal{I}_{F_1}^{l,2}=\mathcal{I}_{L_1}^{l,1}$ and $\mathcal{I}_{F_2}^{l,2}=\mathcal{I}_{L_2}^{l,1}$ by Definition \ref{def LPSFO}), it can decode $\mathcal{I}_{L_1}^{l,2}$ from the transmission $W_1^{l,2} = \mathcal{I}_{F_1}^{l,2} \oplus \mathcal{I}_{F_2}^{l,2} \oplus \mathcal{I}_{L_1}^{l,2}$.  For $i>1$, since $\mathcal{C}_{i}^l$ knows $\mathcal{I}_{F_1}^{l,i+1}$ through overlap with $\mathcal{C}_{i+1}^l$ ($\mathcal{I}_{F_1}^{l,i+1}=\mathcal{I}_{L_1}^{l,i}$ by Definition \ref{def LPSFO}), it can decode $\mathcal{I}_{L_1}^{l,i+1}$ from the transmission $W_{1}^{l,i+1} = \mathcal{I}_{F_1}^{l,i+1} \oplus \mathcal{I}_{L_1}^{l,i+1}$. 
Additionally for $i>2$, from transmissions of the form $W_{j+1}^{l,i+1} = \mathcal{I}_{F_1}^{l,i+1} \oplus \mathcal{I}_{U_j}^{l,i+1}$, each client  $\mathcal{C}_i^l$ with $3 \leq i \leq r_{\max}^l-3$, decodes each $\mathcal{I}_{U_j}^{l,i+1}$ for $j \in [1, i-3]$. A forward decoding chain  occurs then: by the LPS-FO structure, $\mathcal{I}_{F_1}^{l,i+2} = \mathcal{I}_{L_1}^{l,i+1}$, which $\mathcal{C}_i^l$ now possesses. This enables $\mathcal{C}_i^l$ to decode from $\mathcal{C}_{i+2}^l$'s transmissions, obtaining $\mathcal{I}_{L_1}^{l,i+2}$ and $\mathcal{I}_{U_j}^{l,i+2}$. This process continues recursively for all clients $k$ with $i < k \leq r_{\max}^l-2$, allowing $\mathcal{C}_i^l$ to accumulate the set $\{\mathcal{I}_{L_1}^{l,k} \mid k \in [i+1, r_{\max}^l-2]\} \cup \{\mathcal{I}_{U_j}^{l,k} \mid k \in [i+1, r_{\max}^l-2], j \in [1, k-3]\}$ from subsequent clients.

Simultaneously, $\mathcal{C}_i^l$, for $i \in [4, r_{\max}-1]$ engages in backward decoding from preceding clients. Client $\mathcal{C}_i^l$ knows $\mathcal{I}_{L_1}^{l,i-1}$ (which overlaps with its own side-information set, i.e., $\mathcal{I}_{F_1}^{l,i}=\mathcal{I}_{L_1}^{l,i-1}$). From the transmission $W_{1}^{l,i-1} = \mathcal{I}_{F_1}^{l,i-1} \oplus \mathcal{I}_{L_1}^{l,i-1}$ by client $\mathcal{C}_{i-1}^l$, it receives $\mathcal{I}_{F_1}^{l,i-1}$, enabling it to further decode $\mathcal{I}_{U_j}^{l,i-1}$ from transmissions $W_{j+1}^{l,i-1} =  \mathcal{I}_{F_1}^{l,i-1} \oplus \mathcal{I}_{U_j}^{l,i-1}$, for $j \in [i-4]$. The backward decoding chain continues  because $\mathcal{I}_{F_1}^{l,i-1} = \mathcal{I}_{L_1}^{l,i-2}$, allowing $\mathcal{C}_i^l$ to decode from $\mathcal{C}_{i-2}^l$, obtaining $\mathcal{I}_{F_1}^{l,i-2}$ and subsequently $\mathcal{I}_{U_j}^{l,i-2}$. This iterative backward decoding extends to all clients $k$ with $3 \leq k < i$, providing $\mathcal{C}_i^l$ with the set $\{\mathcal{I}_{F_1}^{l,k} \mid k \in [2, i-1]\} \cup \{\mathcal{I}_{U_j}^{l,k} \mid k \in [2, i-1], j \in [1, k-2]\}$. Nobody, except client $\mathcal{C}_1^l$, will decode anything from the transmission made by client $\mathcal{C}_2^l.$

The client $\mathcal{C}_2^l$ makes one transmission.
Each intermediate client $\mathcal{C}_i^l$, for $i \in [3, r_{\max}^l-2]$, transmits exactly $i-2$ coded messages during \textbf{Phase~1}, yielding a total of $T^l = 1+\sum_{i=3}^{r_{\max}^l-2} (i-2) = 1+\frac{(r_{\max}^l-3)(r_{\max}^l-4)}{2}$ transmissions for this phase. Let $T_i^{l,1}$ denote the number of new messages decoded by client $\mathcal{C}_i^l$ (where $i \in [1, r_{\max}^l-1]$) from \textbf{Phase~1} transmissions. Since each transmission contains exactly one unknown message that the appropriate client can recover, every client obtains one new message from every transmission it is able to decode. For a transmitting client $\mathcal{C}_i^l$ with $i \in [3, r_{\max}^l-2]$, the $i-2$ transmissions it makes are not decodable by itself, plus it cannot decode from the coded message transmitted by client $\mathcal{C}_2^l$, consequently, it can decode from the remaining $T^l - (i-2)-1$ transmissions. Hence, $T_i^{l,1} = T^l - (i-1)$ for $i \in [3, r_{\max}^l-2]$. The client $\mathcal{C}_2^l$ cannot decode from its own transmission which is one coded message. Hence, $T_2^{l,1} = T^l - 1$. The  non‑transmitting client in this phase, $\mathcal{C}_1^l$ , participates fully in the forward and backward decoding chains, respectively, and therefore it decodes exactly one new message from every \textbf{Phase~1} transmission. Thus, $T_1^{l,1} = T^l$. The  non‑transmitting client  $\mathcal{C}_{r_{\max}^l-1}^l$ also participates fully in the forward and backward decoding chains, except the transmission made by client $\mathcal{C}_2^l$, and therefore it decodes exactly $ T_{r_{\max}^l-1}^{l,1} = T^l-1$ messages. 
After \textbf{Phase~1}, each client $\mathcal{C}_i^l$ with $i \in [1, r_{\max}^l-1]$ has acquired $T_i^{l,1}$ additional messages beyond its initial side‑information set.

\subsection{{\bf Phase 2}: Client $\mathcal{C}_{r_{\max}^l-1}^l$ Transmission Strategy}
\label{subsec:phase2-proof-alg3}

{\bf Phase 2} consists exclusively of transmissions from client $\mathcal{C}_{r_{\max}^l-1}^l$, which transmits exactly $r_{\max}^l - 3$ messages. The first transmission is constructed as $W_{1}^{l,r_{\max}^l-1} = \mathcal{I}_{F_1}^{l,r_{\max}^l-1} \oplus \mathcal{I}_{L_1}^{l,r_{\max}^l-1}$.
The remaining $r_{\max}^l - 4$
transmissions are constructed as $W_{j+1}^{l,r_{\max}^l-1}
= \mathcal{I}_{L_1}^{l,r_{\max}^l-1}
\oplus
\mathcal{I}_{U_{j}}^{l,r_{\max}^l-1}$,
where
$j \in [1, r_{\max}^l - 4]$.

These transmissions help to continue forward decoding chains for clients in the set
$\{\mathcal{C}_i^l, i \in [r_{\max}^l-2]\}$
and introduce backward decoding chain for the client $\mathcal{C}_{r_{\max}^l}^l$.
For any client $\mathcal{C}_i^l$ with $i < r_{\max}^l-1$, knowledge of
$\mathcal{I}_{F_1}^{l,r_{\max}^l-1}$
(obtained from {\bf Phase 1}) enables decoding of
$\mathcal{I}_{L_1}^{l,r_{\max}^l-1}$
from $W_{1}^{l,r_{\max}^l-1}$ and
$\mathcal{I}_{U_{j}}^{l,r_{\max}^l-1}$
from $W_{j+1}^{l,r_{\max}^l-1}$
for $j \in [1, r_{\max}^l - 4]$.


Client $\mathcal{C}_{r_{\max}^l}^l$ can decode
$\mathcal{I}_{F_1}^{l,r_{\max}^l-1}$
from $W_{1}^{l,r_{\max}^l-1}$ as
$\mathcal{I}_{L_1}^{l,r_{\max}^l-1}
= \mathcal{I}_{F_1}^{l,r_{\max}^l}$ is available with it.
So, it can further decode
$\mathcal{I}_{U_j}^{l,r_{\max}^l-1}$
from $W_{j+1}^{l,r_{\max}^l-1}$
for $j \in [1, r_{\max}^l - 4]$,
because it knows
$\mathcal{I}_{F_1}^{l,r_{\max}^l-1}$.

Let $T_{i}^{l,2}$ be the number of messages decoded by the client $\mathcal{C}_i^l$ for $i \in [1, r_{\max}]$ after completing {\bf Phase 2} of this level $l$.
After {\bf Phase 2}, intermediate clients $\mathcal{C}_i^l$ with $i < r_{\max}^l-1$ have decoded an additional $r_{\max}^l - 3$ messages, bringing their total decoded message count to $T_{i}^{l,2}=T_i^{l,1} +(r_{\max}^l - 3)$. Client $\mathcal{C}_{r_{\max}^l-1}^l$, being the transmitter, gains no new messages from its own transmissions, maintaining a total of $T_{r_{\max}^l-1}^{l,2}=T_{r_{\max}^l-1}^{l,1}$ decoded messages. 

Client $\mathcal{C}_{r_{\max}^l}^l$ decodes $r_{\max}^l - 3$ new messages from {\bf Phase 2}. Now, since it has decoded the message $\mathcal{I}_{F_1}^{l,r_{\max}^l-1}$ combining with the fact that $\mathcal{I}_{F_1}^{l,r_{\max}^l-1}=\mathcal{I}_{L_1}^{l,r_{\max}^l-2}$ it enters the backward chain and decodes all the messages in the set $\{\mathcal{I}_{F_1}^{l,k} \mid k \in [3, r_{\max}^l-2]\} \cup \{\mathcal{I}_{U_j}^{l,k} \mid k \in [3, r_{\max}^l-2], j \in [1, k-3]\}$. It cannot decode anything from the transmission made by client $\mathcal{C}_2^l$. The number of elements in the set given above would be equal to one less than the total number of transmissions in {\bf Phase 1}.  This results in a cumulative total of $T_{r_{\max}^l}^{l,2}=T^l+ r_{\max}^l - 4$ decoded messages. 

We know that, $
T^l = \frac{(r_{\max}^l-3)(r_{\max}^l-4)}{2} + 1$,
and that the total number new messages required by $\mathcal{C}_{r_{\max}^l}^l$ is $C^l - r_{\max}^l  + 1$.  Now,

\begin{align}
T_{r_{\max}^l}^{l,2}
&= \frac{(r_{\max}^l-3)(r_{\max}^l-4)}{2}
   + r_{\max}^l - 3 \nonumber \\
&= \frac{({r_{\max}^l})^2-5r_{\max}^l+6}{2}.
\label{a3p2lhs1}
\end{align}

The condition for {\bf{Algorithm 3}} is that $
C^l = \frac{(r_{\max}^l-2)(r_{\max}^l-1)}{2} + 1$.
On calculating $C^l - r_{\max}^l  + 1$, we get,

\begin{align}
C^l - r_{\max}^l  + 1
&= \frac{(r_{\max}^l-2)(r_{\max}^l-1)}{2}
   - r_{\max}^l + 2 \nonumber \\
&= \frac{({r_{\max}^l})^2-5r_{\max}^l+6}{2}.
\label{a3p2rhs1}
\end{align}

Hence, from (\ref{a3p2lhs1}) and (\ref{a3p2rhs1}) we can prove that $T_{r_{\max}^l}^{l,2} = C^l - r_{\max}^l  + 1$.
This brings $\mathcal{C}_{r_{\max}^l}^l$ to its target message count.

\subsection{{\bf Phase 3}: Client $\mathcal{C}_{r_{\max}^l}^l$ Transmission Strategy}
\label{subsec:phase3-proof-alg3}

{\bf Phase 3} completes the delivery process through transmissions exclusively from client $\mathcal{C}_{r_{\max}^l}^l$, which sends exactly $C^l - r_{\max}^l - T^l + 3$ messages, each of the form
$W_j^{l,r_{\max}^l} = \mathcal{I}_{F_1}^{l,r_{\max}^l} \oplus \mathcal{I}_{U_j}^{l,r_{\max}^l}$
for $j \in [1, C^l - r_{\max}^l - T^l + 3]$.
The decodability of these transmissions relies on the receiving clients' knowledge of $\mathcal{I}_{F_1}^{l,r_{\max}^l}$. 
The total number of transmissions done in this phase can be calculated as,

\begin{align}
C^l - r_{\max}^l  - T^l + 3
=& \left (\frac{(r_{\max}^l-2)(r_{\max}^l-1)}{2} + 1 \right ) - r_{\max}^l  \nonumber 
\\& -\left (\frac{(r_{\max}^l-3)(r_{\max}^l-4)}{2} + 1 \right) +3\nonumber \\
=& \frac{{r_{\max}^l}^2-3r_{\max}^l+2}{2}  - r_{\max}^l\nonumber \\
& - \frac{{r_{\max}^l}^2-7r_{\max}^l+12}{2} + 3 \nonumber \\
=& \frac{4r_{\max}^l-10}{2}
   - r_{\max}^l + 3 \nonumber \\
=& 2r_{\max}^l - 5
   - r_{\max}^l + 3 \nonumber \\
=& r_{\max}^l - 2 . 
\label{a3p3}
\end{align}

Under the assumption $K \ge 2P$, each client $\mathcal{C}_i^l$ possesses at least $i-1$ unique messages.
For $i = r_{\max}^l$, this means at least $r_{\max}^l - 1$ unique messages are available for $\mathcal{C}_{r_{\max}}^l$. 
Since the number of unique messages required for transmission in this phase is $r_{\max}^l - 2$, which is clearly less than the number of unique messages available for $\mathcal{C}_{r_{\max}}^l$, all XOR combinations are well-defined and valid.

Let $T_{i}^{l,3}$ be the number of messages decoded by the client
$\mathcal{C}_i^l$ for $i \in [1, r_{\max}^l]$
after completing {\bf Phase 3} of this level $l$.
For intermediate clients $\mathcal{C}_i^l$ with
$i \in [1, r_{\max}^l-2]$,
from {\bf Phase 2}, it has already decoded
$\mathcal{I}_{F_1}^{l,r_{\max}^l}$
since
$\mathcal{I}_{F_1}^{l,r_{\max}^l}
= \mathcal{I}_{L_1}^{l,r_{\max}^l-1}$.
These clients obtained
$\mathcal{I}_{L_1}^{l,r_{\max}^l-1}$
from {\bf Phase 2},
thus they possess
$\mathcal{I}_{F_1}^{l,r_{\max}^l}$.
Consequently, from each transmission
$W_j^{l,r_{\max}^l}
= \mathcal{I}_{F_1}^{l,r_{\max}^l}
\oplus
\mathcal{I}_{U_j}^{l,r_{\max}^l}$,
they decode
$\mathcal{I}_{U_j}^{l,r_{\max}^l}$,
acquiring
$C^l - r_{\max}^l - T^l + 3$
new messages.
Hence, the total messages acquired so far in this level $l$ is
\[
T_{i}^{l,3}
= T_{i}^{l,2}
+ C^l - r_{\max}^l - T^l + 3
= C^l - (i-1).
\]

Client $\mathcal{C}_{r_{\max}^l-1}^l$ naturally knows
$\mathcal{I}_{F_1}^{l,r_{\max}^l}$
as it equals
$\mathcal{I}_{L_1}^{l,r_{\max}^l-1}$,
which is part of its own side-information set.
Therefore, it can directly decode
$\mathcal{I}_{U_j}^{l,r_{\max}^l}$
from each
$W_j^{l,r_{\max}^l}$,
also gaining
$C^l - r_{\max}^l - T^l + 3$
new messages.
Hence, the total messages acquired so far in this level $l$ is
\[
T_{r_{\max}^l-1}^{l,3}
= T_{r_{\max}^l-1}^{l,2}
+ C^l - r_{\max}^l - T^l + 3
= C^l - (r_{\max}^l - 2).
\]

Client $\mathcal{C}_{r_{\max}^l}^l$,
being the transmitter,
gains no additional messages from {\bf Phase 3} transmissions.
Hence, the total messages acquired so far in this level $l$ is
\[
T_{r_{\max}^l}^{l,3}
= T_{r_{\max}^l}^{l,2}
= C^l - (r_{\max}^l - 1).
\]

In short, by the end of Phase~3, each client
$\mathcal{C}_i^l \in \mathcal{C}^l$
has decoded
$T_{i}^{l,3}
= C^l - (i-1)$
messages from the transmissions made at this level.

Since $T = K + C$ by definition, and we have
$K^{l} = K^{l-1} + r_{\max}^{l-1}$
and
$C^{l} = C^{l-1} - r_{\max}^{l-1}$,
it follows that
$K^{l} + C^{l} = K^{l-1} + C^{l-1}$.
By induction,
$K^l + C^l = K + C = T$,
establishing that
$T = K^{l} + C^{l}$.
This confirms that the number of additional messages required by the first client in the set
$\mathcal{C}^{l}$
to reach the target remains exactly
$C^{l}$.
Hence, the number of additional messages required for any client
$\mathcal{C}^{l}_i \in \mathcal{C}^{l}$,
for $i \in [C^l]$,
to reach the target is exactly
$C^{l}- (i-1)$.
We have proved that each client
$\mathcal{C}_i^l \in \mathcal{C}^l$, for $i \leq r_{\max}^l$
has decoded
$T_{i}^{l,3}
= C^l- (i-1)$
messages from the transmissions made at this level.
Hence each of these clients 
are satisfied from the transmissions made at this level.

 Clients $\mathcal{C}_i^l \in \mathcal{C}^l$ with $i > r_{\max}^l$ receive zero messages during the current recursive level.
A crucial property emerges for clients with index $i > r_{\max}^l$. These clients cannot decode any messages from Phase~1 transmissions. All the transmissions $W$ in this level is of the type $W = \mathcal{I}_{F_1}^{l,k} \oplus \mathcal{I}_{F_2}^{l,k} \oplus \mathcal{I}_{L_1}^{l,k},W = \mathcal{I}_{F_1}^{l,k} \oplus \mathcal{I}_{L_1}^{l,k}$, or   $W = \mathcal{I}_{F_1}^{l,k} \oplus \mathcal{I}_{U_j}^{l,k}$, for $k \leq r_{\max}^l$,  for some integer $j$. The  client $\mathcal{C}_i^l$ (where $i > r_{\max}^l$) lacks knowledge of both XORed terms in all these kind of transmissions as they do not have any overlap with any of the messages present in the XOR. No linear combination of received transmissions reduces the number of unknown terms to one as from the transmission structure in each of the phases, any two transmissions have atmost one message in common. Consequently, clients $\mathcal{C}_{i}^l$, for $i>r_{\max}^l$, remain unaffected by the transmissions in this level, preserving their original message sets and enabling the recursive decomposition.



\section{Recursive Structure Preservation}
\label{sec: recursive-structure}

The recursive algorithm maintains the LPS-FO structure at each level, enabling the consistent application of Algorithms~2 and~3 throughout the decomposition. This section establishes the inductive preservation of the LPS-FO properties and verifies the parameter relationships that ensure correctness across recursive levels. The recursive process systematically partitions the original client set into successive layers. At each level $l$, the algorithm focuses exclusively on the first $r_{\max}^l$ clients in the current set $\mathcal{C}^l$, employing Algorithms~2 or~3 to elevate them to the target knowledge level $T$. The remaining clients are intentionally shielded from receiving any messages during this level, ensuring their knowledge states remain unaltered (from the proofs in Section \ref{sec:proof-alg2} and \ref{sec:proof-alg3}). After processing these $r_{\max}^l$ clients, they are removed from consideration, and the next $r_{\max}^{l+1}$ clients become the focus at level $l+1$. This sequential approach guarantees that at any given level, only the clients currently being processed are affected, while all others retain their original side-information.

This decomposition strategy inherently preserves the LPS-FO structure because each subset $\mathcal{C}^l$ constitutes a contiguous segment of the original client sequence. The linear progression and fixed overlap properties are local characteristics that remain valid within any contiguous subsequence. Consequently, the algorithm can be applied recursively without modification, relying solely on the same transmission schemes (Algorithms~2 and~3) at every level.

\subsection{Initial Level ($l = 1$)}

At the initial recursive level $l = 1$, the original set of clients $\mathcal{C}^1 = \{\mathcal{C}_1, \mathcal{C}_2, \ldots, \mathcal{C}_C\}$ possesses the LPS-FO structure by problem definition. The side-information sets satisfy $|\mathcal{I}_i| = K + (i-1)$ for $i \in [1, C]$, with consecutive clients sharing exactly $P$ overlapping messages. After executing either Algorithm~2 or~3 at level $1$, the first $r_{\max}^1$ clients in $\mathcal{C}^1$ attain exactly $T = K + C$ messages (proved in Section \ref{sec:proof-alg2} and \ref{sec:proof-alg3}) and cease transmission. Crucially, no other clients receive any new messages during this level, as proven in Section \ref{sec:proof-alg2} and \ref{sec:proof-alg3}. The remaining clients $\{\mathcal{C}_i \mid i \in [r_{\max}^1+1, C]\}$ retain their original side-information sets unchanged.

\subsection{Inductive Step: From Level $l$ to Level $l+1$}

Assume that at recursive level $l$, we have a set of clients $\mathcal{C}^l = \{\mathcal{C}_1^l, \mathcal{C}_2^l, \ldots, \mathcal{C}_{C^l}^l\}$ that maintains the LPS-FO structure with parameters $K^l$ and $P$, where $K^l$ denotes the number of messages in the side-information set of the first client $\mathcal{C}_1^l$. The cardinality of side-information set of client $\mathcal{C}_i^l$ is $K^l + (i-1)$ for $i \in [1, C^l]$, with consecutive clients sharing exactly $P$ overlapping messages. 

After applying Algorithm~2 or~3 at level $l$, the first $r_{\max}^l$ clients in $\mathcal{C}^l$ achieve exactly $T = K + C$ messages and stop transmitting. The remaining $C^l - r_{\max}^l$ clients receive no new messages during this level, preserving their original knowledge states. We then construct the client set for level $l+1$ as $\mathcal{C}^{l+1} = \{\mathcal{C}_1^{l+1}, \mathcal{C}_2^{l+1}, \ldots, \mathcal{C}_{C^{l+1}}^{l+1}\}$, where $\mathcal{C}_i^{l+1} = \mathcal{C}_{r_{\max}^l + i}^l$ for $i \in [1, C^{l+1}]$ and $C^{l+1} = C^l - r_{\max}^l$.

We now verify that $\mathcal{C}^{l+1}$ inherits the LPS-FO structure. The first client $\mathcal{C}_1^{l+1} = \mathcal{C}_{r_{\max}^l+1}^l$ originally possessed $K^l + r_{\max}^l$ messages. Denoting $K^{l+1} = K^l + r_{\max}^l$, we observe that for any client $\mathcal{C}_i^{l+1} = \mathcal{C}_{r_{\max}^l+i}^l$, the cardinality of side-information set size is $ K^l + (r_{\max}^l + i - 1) = (K^l + r_{\max}^l) + (i-1) = K^{l+1} + (i-1)$. Thus, the linear progression property holds. Moreover, the fixed overlap of $P$ messages between consecutive clients is preserved because the relative ordering and overlapping structure of the clients remain unchanged. Therefore, $\mathcal{C}^{l+1}$ satisfies the LPS-FO conditions with parameters $K^{l+1}$ (the cardinality of the side-information set of the first client in $\mathcal{C}^{l+1}$) and $P$ (the overlap parameter).

\section{Base Cases Verification}
\label{sec:base-cases}

The recursive algorithm terminates when the number of remaining clients at level $l$, denoted $C^l$, is reduced to 2 or 1. These base cases are handled by explicit transmission schemes that ensure the final clients reach the target knowledge level $T = K + C$. Crucially, these transmissions are designed so that no other clients in the system can decode additional messages, maintaining the security constraint that each client obtains exactly $T$ messages.

\subsection{Two-Client Base Case ($C^l = 2$)}

When only two clients remain at recursive level $l$, they correspond to the final two clients in the original ordering, specifically $\mathcal{C}_1^l = \mathcal{C}_{C-1}$ and $\mathcal{C}_2^l = \mathcal{C}_C$. The transmission scheme consists of three coded broadcasts. First, client $\mathcal{C}_2^l$ transmits $W_1^l = \mathcal{I}_{F_1}^{l,2} \oplus \mathcal{I}_{U_1}^{l,2}$ and $W_2^l = \mathcal{I}_{F_1}^{l,2} \oplus \mathcal{I}_{U_2}^{l,2}$. Client $\mathcal{C}_1^l$, possessing $\mathcal{I}_{F_1}^{l,2}$ through the fixed overlap with $\mathcal{C}_2^l$, decodes both $\mathcal{I}_{U_1}^{l,2}$ and $\mathcal{I}_{U_2}^{l,2}$, acquiring two new messages. Subsequently, client $\mathcal{C}_1^l$ transmits $W_3^l = \mathcal{I}_{L_1}^{l,1} \oplus \mathcal{I}_{U_1}^{l,1}$. Client $\mathcal{C}_2^l$, knowing $\mathcal{I}_{L_1}^{l,1}$ through the overlap with $\mathcal{C}_1^l$, decodes $\mathcal{I}_{U_1}^{l,1}$, gaining one new message. After these exchanges, both clients have received precisely the number of additional messages needed to reach $T = K + C$ messages in total (which is proved in Section ). 

For all other clients in the system (those with indices $i < C-1$), these transmissions yield no new messages. Consider any client $\mathcal{C}_i$ with $i < C-1$. The transmissions $W_1^l$ and $W_2^l$ both involve $\mathcal{I}_{F_1}^{l,2}$ , $\mathcal{I}_{U_j}^{l,2}$ (for $j \in \{1,2\}$), which are unknown to $\mathcal{C}_i$ since $\mathcal{I}_{F_1}^{l,2}$ belongs exclusively to the overlapping region between $\mathcal{C}_{C-1}$ and $\mathcal{C}_C$ and $\mathcal{I}_{U_j}^{l,2}$  belongs exclusively to  $\mathcal{C}_{C}$. Similarly, $W_3^l$ involves $\mathcal{I}_{L_1}^{l,1}$, $\mathcal{I}_{U_1}^{l,1}$, which are unknown to $\mathcal{C}_i$ since $\mathcal{I}_{F_1}^{l,1}$ belongs exclusively to the overlapping region between $\mathcal{C}_{C-1}$ and $\mathcal{C}_C$ and $\mathcal{I}_{U_1}^{l,1}$  belongs exclusively to  $\mathcal{C}_{C-1}$. which is also unknown to $\mathcal{C}_i$ as it lies in the overlap between $\mathcal{C}_{C-1}$ and $\mathcal{C}_{C}$. Also, none of these messages were involved in transmissions done at any previous recursive levels.
Since each transmission is an XOR of two unknown messages, no client other than the intended recipient can decode any new information. Also, no linear combination of received transmissions can reduce the number of unknown terms to one as there are 3 transmissions and each transmission has 1 new message in the XOR sum. This preserves the security condition, preventing unintended knowledge acquisition.

\subsection{Single-Client Base Case ($C^l = 1$)}

When only one client remains at recursive level $l$, this client corresponds to the final client $\mathcal{C}_C$ in the original ordering. In this scenario, client $\mathcal{C}_{C-1}^1$ (from the original indexing) transmits a single coded message $W_1^l = \mathcal{I}_{L_1}^{1,C-1}  \oplus \mathcal{I}_{U_L}^{1,C-1}$. Client $\mathcal{C}_1^l$ (which is $\mathcal{C}_C$ in original indexing) knows $\mathcal{I}_{L_1}^{1,C-1}$ through the fixed overlap with $\mathcal{C}_{C-1}$, enabling it to decode $\mathcal{I}_{U_L}^{1,C-1}$, thereby acquiring the final message required to reach $T = K + C$ messages (which is proved in Sections \ref{sec:proof-alg2} and \ref{sec:proof-alg3}). 

For all other clients $\mathcal{C}_i$ with $i < C-1$, this transmission does not contain any decodable information. In particular, the message $\mathcal{I}_{L_1}^{1,C-1}$ is unknown to these clients since it lies in the intersection of $\mathcal{C}_{C-2}$ and $\mathcal{C}_{C-1}$, while $\mathcal{I}_{U_L}^{1,C-1}$ is unique to $\mathcal{C}_{C-1}$.

We now show that $\mathcal{I}_{U_L}^{1,C-1}$ does not appear in any transmission made by client $\mathcal{C}_{C-1}$ at the previous level $l-1$. At level $l-1$, client $\mathcal{C}_{C-1}$ corresponds to $\mathcal{C}_{r_{\max}^{l-1}}^{\,l-1}$. Hence, the only transmissions that need to be examined are the Phase~3 transmissions at level $l-1$.

If \textbf{Algorithm~3} is executed at level $l-1$, the number of unique messages required in Phase~3 is $r_{\max}^{l-1} - 2$. The client $\mathcal{C}_{r_{\max}^{l-1}}^{\,l-1}$ has at least $r_{\max}^{l-1} - 1$ messages, it follows that  the last unique message is not used in Phase~3, and $\mathcal{I}_{U_L}^{1,C-1}$ remains uninvolved.

If \textbf{Algorithm~2} is executed at level $l-1$, Phase~3 requires $r_{\max}^{l-1} - 1$ unique messages. From the LPS-FO side-information structure, when $l>3$, client $\mathcal{C}_{C^1}^{l-1}$ has at least $r_{\max}^{l-1} - 1$ messages. Since $r_{\max}^{l-1} \le r_{\max}^{l-2}$ and $r_{\max}^1 > 2$, the side-information size of the first client in level $l-1$, namely $\mathcal{C}_1^{l-1}$, exceeds $K$. Consequently, $K^{l-1} > K$, which implies that client $\mathcal{C}_{C^{l-1}}^{l-1}$ has at least $r_{\max}^{l-1}$ messages, out of which only $r_{\max}^{l-1} - 1$ messages are used in that level $l-1$. Hence, the last unique message is not used in Phase~3, and $\mathcal{I}_{U_L}^{1,C-1}$ remains uninvolved.
Thus, the only remaining case to consider is $l = 2$. By exhaustive inspection, the only configuration in which the base case propagates to level $2$ occurs when $C = 5$. In this case, client $\mathcal{C}_{C-2}$ (the second-to-last client) uses only two messages during Phase~3 of \textbf{Algorithm~2} in level $1$, leaving at least one unique message unused. Hence, even in this scenario, the final unique message does not appear in any prior transmissions.

Since both XORed components are unknown to unintended clients, no client can extract any new message from this transmission. Therefore, the security constraint remains satisfied.

The base cases thus complete the recursive algorithm, ensuring that when the process terminates, every client possesses exactly $T = K + C$ distinct messages, with no client exceeding this target. The explicit transmission schemes for $C^l = 2$ and $C^l = 1$, combined with their inherent security properties, provide the necessary termination conditions for the recursive decomposition implemented in Algorithm~1.

\section{Discussion}

In this work, we studied secure DPIC under heterogeneous side-information conditions by adopting the LPS-FO model, which captures the practical non-uniformity in client knowledge. We proposed a transmission scheduling algorithm that efficiently coordinates client broadcasts to maximize coding gain per transmission while ensuring that all clients reach a common target knowledge level. Importantly, the design enforces a strict security constraint, guaranteeing that no client acquires more than a prescribed target number of messages. Our proposed scheme is optimal only when $C=3$ or $4$. As the next step, we aim to further reduce the gap between the achievable number of transmissions and the corresponding lower bound, which is currently $ N(C - r_{\max})$, moving toward transmission-optimal secure decentralized schemes. 

\bibliographystyle{IEEEtran}
\bibliography{ref}

\end{document}